%% file: main.tex
\documentclass[journal,comsoc]{IEEEtran}
\usepackage[utf8]{inputenc}
\usepackage[margin=1in]{geometry}
\usepackage{graphicx}
\usepackage{placeins}
\usepackage{float}
\usepackage{bbm}
\usepackage{subfigure}
\usepackage{verbatim}
\usepackage{amsmath,amssymb,amsfonts, amsthm}
\usepackage{mathtools}
\usepackage[mathscr]{euscript}
\usepackage{mwe}
\usepackage[ruled,vlined]{algorithm2e}
\usepackage{setspace}
\usepackage{cite}
\usepackage{xcolor}

\DeclareMathOperator{\E}{\mathbb{E}}

\DeclareMathOperator*{\argmax}{arg\,max}
\DeclareMathOperator*{\argmin}{arg\,min}

\newtheorem{prop}{Proposition}
\newtheorem{definition}{Definition}
\newtheorem{remark}{Remark}

\include{notation}

\usepackage[justification=centering]{caption}
\usepackage[font=small]{caption}
\newcommand{\subparagraph}{}
\usepackage[explicit]{titlesec}

\let\OLDthebibliography\thebibliography
\renewcommand\thebibliography[1]{
	\OLDthebibliography{#1}
	\setlength{\parskip}{0pt}
	\setlength{\itemsep}{0pt minus 0.0ex}
}

\makeatletter
\def\thickhline{%
	\noalign{\ifnum0=`}\fi\hrule \@height \thickarrayrulewidth \futurelet
	\reserved@a\@xthickhline}
\def\@xthickhline{\ifx\reserved@a\thickhline
	\vskip\doublerulesep
	\vskip-\thickarrayrulewidth
	\fi
	\ifnum0=`{\fi}}
\makeatother

\titlespacing{\section}{4pt}{1ex}{0.25ex}
\titlespacing{\subsection}{0pt}{0.25ex}{0.25ex}
\titlespacing{\subsubsection}{0pt}{0.25ex}{0.25ex}

\newlength{\thickarrayrulewidth}
\title{Constrained Contextual Bandit Learning for Adaptive Radar Waveform Selection}
\begin{document}
\author{Charles E. Thornton, R. Michael Buehrer, and Anthony F. Martone
	
\thanks{\footnotesize C.E. Thornton and R.M. Buehrer are with Wireless @ Virginia Tech, Department of ECE, Blacksburg, VA, USA, 24061. \textit{(Emails: $\{$thorntonc, buehrer$\}$@vt.edu).}\\ A.F. Martone is with the US Army Research Laboratory, Adelphi MD, USA, 20783. \textit{(Email: anthony.f.martone.civ@mail.mil).}\\ The support of the US Army Research Office (ARO) is gratefully acknowledged. Portions of this work were presented at IEEE GLOBECOM, Taipei, Taiwan, Dec. 2020 \cite{prelim}. Other portions will be presented at IEEE Radar Conf., Atlanta, GA, May 2021 \cite{prelim2}.}}
\maketitle	

\IEEEaftertitletext{\vspace{.002\baselineskip}}
\begin{abstract}
A sequential decision process in which an adaptive radar system repeatedly interacts with a finite-state target channel is studied. The radar is capable of passively sensing the spectrum at regular intervals, which provides side information for the waveform selection process. The radar transmitter uses the sequence of spectrum observations as well as feedback from a collocated receiver to select waveforms which accurately estimate target parameters. It is shown that the waveform selection problem can be effectively addressed using a linear contextual bandit formulation in a manner that is both computationally feasible and sample efficient. Stochastic and adversarial linear contextual bandit models are introduced, allowing the radar to achieve effective performance in broad classes of physical environments. Simulations in a radar-communication coexistence scenario, as well as in an adversarial radar-jammer scenario, demonstrate that the proposed formulation provides a substantial improvement in target detection performance when Thompson Sampling and EXP3 algorithms are used to drive the waveform selection process. Further, it is shown that the harmful impacts of pulse-agile behavior on coherently processed radar data can be mitigated by adopting a time-varying constraint on the radar's waveform catalog.
\end{abstract}
\section{Introduction}
Since the inception of radar systems, the fundamental question of ``which waveforms to transmit and when?" has remained largely unresolved, primarily due to the open-ended nature of the problem. In the 1950's, Woodward introduced the \emph{ambiguity function} to highlight the limitation of a particular waveform in jointly measuring delay and Doppler information, thereby demonstrating the need for application specific trade-offs in the design of waveform-receiver pairs \cite{Woodward1953}. Limitations in hardware historically limited radar systems to a fixed transmit waveform, which restricted particular deployments to a relatively static purpose. More recently, advances in hardware and algorithms have allowed for modern radars to adaptively select from a set of diverse waveforms to improve target detection and estimation characteristics. Moreover, since wireless systems are becoming increasingly interference limited \cite{intfLarge}, adaptive waveform selection is of growing importance for radar systems operating in congested spectrum.

\begin{figure}[t]
	\centering
	\includegraphics[scale=0.55]{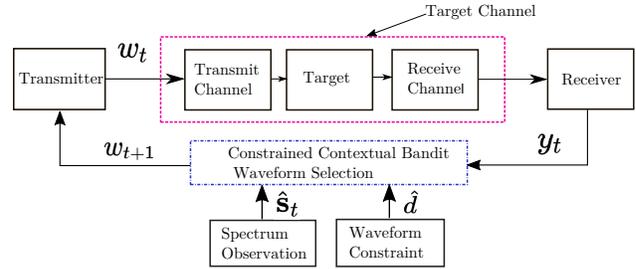}
	\caption{\textsc{System Model.} Role of the proposed waveform selection algorithm in a general closed-loop radar system.}
	\label{fig:block}
	\vspace{-.4cm}
\end{figure} 

Radar waveform selection has traditionally been formulated as a deterministic control or optimization problem \cite{Kershaw1994}. Techniques of this nature require accurate estimates of the noise, clutter, and interference processes to determine optimal sensor-processor parameters for the tracking, detection, or estimation problem of interest. In dynamic interference-limited scenarios, it is often impractical to obtain accurate estimates of the physical environment in real-time. Thus, future radar systems must have the capability to adapt waveform parameters with minimal assumptions regarding the target channel and must make use of limited observations when selecting waveforms. 

\begin{table*}[]
	\small
	\caption{\textsc{Comparison of Waveform Selection Works}}
	\label{tab:lit}
	\begin{tabular}{l|l|l|l|l|l|}
		\cline{2-6}
		& Kershaw/Evans {[}5{]} & Niu et al. {[}14{]} & Sowelam {[}16{]} & \begin{tabular}[c]{@{}l@{}}RL Approaches\\ {[}6{]}/ {[}38{]}/{[}39{]}\end{tabular} & Our Work \\ \cline{2-6} 
		\begin{tabular}[c]{@{}l@{}}Selection \\ Mechanism\end{tabular} & \begin{tabular}[c]{@{}l@{}}Waveform-\\ specific CRLB\end{tabular} & \begin{tabular}[c]{@{}l@{}}Target \\ Kinematic\\ Model\end{tabular} & \begin{tabular}[c]{@{}l@{}}KL\\ information\\ number (for \\ classification)\end{tabular} & \begin{tabular}[c]{@{}l@{}}Available \\ spectrum \\ utilization\end{tabular} & \begin{tabular}[c]{@{}l@{}}Spectrum \\ utilization/\\ waveform \\ distortion\\  constraint\end{tabular} \\ \cline{2-6} 
		\begin{tabular}[c]{@{}l@{}}Target Channel \\ Assumptions\end{tabular} & \begin{tabular}[c]{@{}l@{}}Target always\\ detected,\\ Gaussian \\ measurement errors\end{tabular} & \begin{tabular}[c]{@{}l@{}}Gaussian \\ noise model\end{tabular} & \begin{tabular}[c]{@{}l@{}}Target \\ reflectivity \\ is Gaussian\end{tabular} & None & None \\ \cline{2-6} 
		\begin{tabular}[c]{@{}l@{}}Computational \\ Complexity\end{tabular} & Low (when tractable) & Low & Low & High & Low \\ \cline{2-6} 
		Context & \begin{tabular}[c]{@{}l@{}}Assumes SNR \\ is known\end{tabular} & \begin{tabular}[c]{@{}l@{}}Assumes \\ Kinematic \\ Model\end{tabular} & \begin{tabular}[c]{@{}l@{}}History of\\ received\\ signals\end{tabular} & \begin{tabular}[c]{@{}l@{}}Interference \\ state\end{tabular} & \begin{tabular}[c]{@{}l@{}}Interference \\ estimate/\\ history of \\ received signals\end{tabular} \\ \cline{2-6} 
		Online? & Yes & Yes & Yes & No & Yes \\ \cline{2-6} 
	\end{tabular}
\end{table*}

This work addresses the need for sample efficient and computationally feasible learning schemes for dynamic waveform selection in a closed-loop radar system. Particularly, linear \emph{contextual bandit} algorithms are introduced to guide an adaptive radar's waveform selection process. A visualization of the proposed learning approach in the context of a general radar system is seen in Figure \ref{fig:block}. The significance of this formulation lies in its \textit{context-aware} nature, which allows for the statistical patterns of the surrounding environment to be directly associated with specific waveforms. Further, this model does not assume that the target channel's state sequence is impacted by the radar's choice of waveforms, which is implicitly assumed in related Markov decision process (MDP) formulations \cite{Thornton2020}. Another issue addressed by this work is that of harmful sidelobe levels due to pulse-agility within the radar's coherent processing interval (CPI) \cite{Levanon2004}. We introduce a distance metric between adjacent transmitted waveforms to limit distortion effects in the resulting range-Doppler images. Finally, the proposed learning framework's linear structure provides a scalable model which can be applied to large problems with robust guarantees.

\subsection{Background and Related Work}
\label{ss:bg}
The role of active radio sensing in society is rapidly growing. In the commercial sector, emerging applications include automotive control, human gesture recognition and indoor localization, along with many others \cite{advances}. In military systems, radars are ubiquitously used to provide precise information about both physical targets and the electromagnetic characteristics of devices for strategic purposes. To meet growing performance requirements, modern radar systems must adapt to specific target scenarios and reliably mitigate interference. In particular, radar interference mitigation is expected to become a significant challenge due to the proliferation of communication devices in frequency ranges historically allocated for radar use \cite{technical,Zheng2019}.

Emerging joint communication and sensing applications also require adaptive radar strategies to enable interplay between RF systems that may operate in close proximity. While radar, communications, and electronic warfare systems have traditionally been treated separately, there has been a recent surge of interest in developing multi-function RF systems \cite{Liu2020}. Due to the diverse range of scenarios in which future radar systems will be deployed, \emph{context-aware} intelligent strategies are of particular interest to ensure accurate and reliable sensing performance.

The problem of optimal radar waveform selection has been the subject of much research interest, often approached from either a control-theoretic or information theoretic perspective. A brief comparison of the current contribution with some of these works is presented in Table \ref{tab:lit}. Kershaw and Evans explored a control-theoretic optimization approach to adaptive waveform selection by considering the effect of sensor parameters on tracking performance in a closed-loop system \cite{Kershaw1994}. Alternatively, Bell took an information theoretic view, and proposed radar waveform design based on an information theoretic view of the transmission-reception process. Bell's approach maximizes the SNR at the receiver's matched filter to design waveforms which optimize detection, and maximizes the mutual information between the target and received waveform to design optimal estimation waveforms \cite{Bell1993}. These seminal works are indicative of a large body of work on optimal waveform selection, with many works utilizing the tracking error as an optimality criterion \cite{Hurtado2008,Niu2002,Hong2005}, and many others using information theoretic quantities to guide the waveform selection \cite{Sowelam2000,Tang2019,Zhu2017,Tang2015}. A notable recent development in the information-theoretic perspective is the work of Devroye, which formulates the waveform selection process as a non-classical joint source-channel coding problem with ties to uncoded communication systems \cite{Devroye}. 

In recent years, a body of literature on closed-loop radar systems, often referred to as \textit{cognitive radar}, has emerged. These systems exploit real-time feedback between the transmitter and receiver to adjust transmission and processing parameters, following the view of \cite{Kershaw1994}. A survey of recent advances in this domain can be found in \cite{ontheroad} and the references therein. Many recent contributions have focused on Bayesian particle filtering or optimal beamforming strategies to improve the accuracy of target tracking systems \cite{framework,cogbeamforming}, along with other important applications including adaptive detection \cite{Metcalf2015}, antenna selection \cite{Elbir2019}, and resource management \cite{Mahdi2018}. Cognitive radar techniques have also been applied to facilitate coexistence and interoperability, where spectrum information is utilized in the decision-making process \cite{sharing,Thornton2020,metacog,xampling}. These adaptive strategies can be both practical and effective given that a radar's transmitter and receiver are often co-located, allowing for quick feedback. This structure also allows for statistics about the interference channel to be obtained in real-time via fast spectrum sensing techniques \cite{Martone2018,xampling}.

This work considers interference mitigation between a radar and an unknown shared channel, which relates to the literature on radar and communications coexistence. While a comprehensive survey is beyond the current scope, a summary of recent advances can be found in \cite{Zheng2019} and references therein. Among recent works, many contributions aim to mitigate mutual interference by adapting the communications system transmitter or through joint system planning. Examples include closed-form precoder design \cite{mimocell}, joint optimization of transmit waveforms \cite{convex,codesign2016}, intrapulse radar-embedded communications \cite{Ciuonzo2015}, and estimation of interference channel state information \cite{semiblind,howmuch}. However, co-design is often expensive or impractical for logistical reasons. Furthermore, radar-to-radar interference is an important design consideration in emerging short-range applications \cite{mosarim}, emphasizing the need for non-cooperative spectrum management by adapting the radar's transmitter. 

To adaptively improve the performance of radar systems in dynamic conditions, reinforcement learning (RL) approaches have been proposed \cite{selvi,Thornton2020,PLiu2020}. These RL approaches aim to optimize transmission parameters with respect to a specifically designed objective function using a history of feedback information. This has been achieved by modeling the waveform selection problem as a Markov Decision Process (MDP) \cite{selvi}. However, the application of RL to high-dimensional problems encountered in the real-world often requires an extensive period of offline exploration to learn an effective transmission policy. An extended period of suboptimal behavior may not be practical for time-sensitive applications such as target tracking, spectrum sharing, or electronic warfare. Additionally, the complexity of RL approaches, such as dynamic programming or \textit{Q}-learning, can quickly become intractable for realistic applications as the size of the state-action space that defines the problem increases \cite{sutton}. Further, while RL techniques generally perform well in cases where the environment obeys the Markov property \cite{selvi}, interference in wireless networks is often a time-varying stochastic process, where the Markov property is not guaranteed to hold. Another practical issue associated with an MDP formulation is that the decision-maker's actions are assumed to affect the future states of the scene. In general, this assumption may not hold for environments where other systems and the target may have properties that are independent of the radar's transmissions. 

Thus, there are several significant challenges in directly applying the RL framework to cognitive radars. To mitigate the computational issues and modeling assumptions encountered when solving for an optimal policy in the full RL problem, which considers state transitions as a stochastic process, multi-armed bandit (MAB) approaches are often used for sequential decision making under uncertainty due to their simplicity and theoretical performance guarantees \cite{slivkins}. MAB approaches have shown great promise in developing a variety of spectrum access strategies for cognitive radio link adaptation and channel selection \cite{dsa1,dsa2}. A generalization of the MAB problem which considers side information each decision round is known as a \textit{contextual bandit} (CB) and has been the subject of intense investigation \cite{context,slivkins}. In this work, we utilize linear variants of the CB formulation to develop a model of the radar's waveform selection process. These models provide a flexible approach that can be applied to deterministic, stochastic, and adversarial environments of practical interest.

\subsection{Contributions}
The main contributions of this paper are the following:

\begin{itemize}	
	\item Expanding on previous waveform selection strategies, which either deterministically select a waveform \cite{Kershaw1994}, rely upon a restrictive set of assumptions about the target channel \cite{Bell1993}, or require significant experience to achieve good performance \cite{Thornton2020}, we model the closed-loop waveform selection process as a sequential decision process using side information obtained from spectrum sensing. Performance measures are discussed and a cost function that considers interference avoidance and available bandwidth utilization is developed. The formulation also considers mitigation of harmful distortion effects due to pulse-agility by implementing a time-varying constraint on the waveform catalog.
	
	\item The problem is addressed in a computationally feasible manner using a linear contextual bandit online learning framework. Stochastic and adversarial models are introduced, which allow the radar to perform reliably across a wide range of scenarios while making minimal assumptions regarding the physical environment. Thompson Sampling (TS) and EXP3 algorithms which provide near-optimal performance in the respective stochastic and adversarial settings are described.
	
	\item A simulation study involving radar-communications coexistence scenarios as well as radar-adaptive jammer adversarial scenarios is performed. The proposed algorithms are shown to be sample efficient, and result in improved detection and tracking performance relative to a non-adaptive radar and a simple reactive strategy.
	
	\item Considerations for tracking a single target are discussed, and selection of optimal waveform parameters using information from a Kalman tracking system to reduce the decision space of the contextual bandit algorithm is described. In simulation, improved tracking performance is shown compared to a static bandwidth allocation and a naive reactive strategy.
\end{itemize}

\subsection{Notation}
The following notation is used. Bold upper (and lower) case letters denote matrices (and vectors) $\mathbf{X}$ (and $\mathbf{x}$). $\mathbf{X}^{T}$ is the transpose operation. Upper case script letters, such as $\mathcal{A}$, denote sets. $\mathbf{I}_{d}$ is the $d \times d$ identity matrix. $\mathbf{0}_{d}$ is a length $d$ vector of zeros. $\langle \cdot, \cdot \rangle$ is the inner product operation. $\lVert \cdot \rVert$ is the $\ell_{2}$-norm. $\mathbb{P}(\cdot)$ is a probability measure on a measurable space $(\Omega,\mathcal{F})$. $\mathbb{R}$ and $\mathbb{N}$ denote the sets of all real and natural numbers, respectively. $\mathbbm{1}\{\cdot\}$ is the indicator function, which returns $1$ if the argument is true and $0$ otherwise. $\mathcal{O}(\cdot)$ is the standard Bachmann-Landau notation and $\tilde{\mathcal{O}}(\cdot)$ similarly describes asymptotic behavior while ignoring logarithmic terms.

\subsection{Organization}
The remainder of the paper is organized as follows. In Section \ref{se:form}, the adaptive waveform selection problem is described mathematically, along with the challenges that highlight the need for sample-efficient online learning. Section \ref{se:online} discusses the stochastic and adversarial linear bandit formulations, as well as the proposed TS and EXP3 algorithms to solve the problem efficiently. In Section \ref{se:track}, relevant considerations for the application of the online learning approach to single-target tracking are discussed. In Section \ref{se:sim}, the proposed algorithms are extensively evaluated coexistence and an adaptive jamming simulations. Section \ref{se:concl} provides concluding remarks.

\section{Problem Formulation}
\label{se:form}
\subsection{System Model}
Consider a stationary and monostatic radar system located at the origin. The radar participates in a sequential decision process in which time is slotted into a sequence of discrete intervals $t = 1,2,...,n$, and a waveform must be selected at each step. Each time index $t \in \mathbb{N}$ thus corresponds to the radar's $t^{\text{th}}$ pulse repetition interval (PRI). The radar wishes to detect targets in the physical environment, or \emph{scene}, and measure a time-evolving random vector of target parameters $\mathbf{z}_{t} \in \mathcal{Z}$. The radar operates in a shared channel with center frequency $f_{\texttt{c}}$ and bandwidth $B$. The shared channel may contain one or more communication systems, whose transmission strategy is unknown to the radar \textit{a priori} and may result in mutual interference if there is temporal and spectral overlap with the radar's transmission. Within each PRI, the radar must select a linear frequency modulated (LFM) chirp waveform\footnote{This approach is applicable to broader waveform catalogs, which may include other pulse modulation schemes such as FM noise, phase coded, or non-linear FM waveforms, but we focus on LFM here for simplicity and favorable ambiguity properties in measuring both range and Doppler information.} $w_{i}$ from a finite indexed catalog $\mathcal{W} = \{ w_{i} \}_{i = 1}^{W}$. The time-domain transmitted waveform is given by \cite{skolnik}
\begin{equation}
w_{i}(t^{\prime}) = A \; \exp(-t'^{2} / T^{2}) \cos{(2 \pi f_{i} t^{\prime} + \pi \alpha_{i} t^{\prime 2})},
\label{eq:tx}
\end{equation}
where $t^{\prime}$ corresponds to continuous \textit{`fast time'} within a PRI, $A$ is a constant amplitude, $T$ is the pulse duration, $f_{i}$ is the carrier frequency, and $\alpha_{i}$ is the slope of the up-chirp frequency, which dictates the signal bandwidth, given by $\texttt{BW}_{i} = T \alpha_{i}$. It is assumed that the narrowband property $\texttt{BW}_{i} << f_{\texttt{i}}$, holds for each $w_{i} \in \mathcal{W}$. A reasonable objective for the radar is to select the sequence of waveforms $\{w_{1}, w_{2},...,w_{n}\}$ which results in accurate and unambiguous estimation of the random target parameters of interest, represented by the random vector $\mathbf{z}_{t} \in \mathcal{Z}$. However, the radar is operating in a time-varying and interference-limited environment, so it must utilize information about the current state of the channel to select a waveform at each PRI. Thus, the waveform selection process can be modeled as a sequential decision process with side information. To precisely define the problem, we now examine the underlying randomness in the target channel and formulate the radar's spectrum sensing process.

\subsection{Target Channel Model and Spectrum Sensing Process}
The selected waveform at the $t^{\text{th}}$ PRI, $w_{t}$, is an input to a target channel. The target channel is a stochastic mapping between transmitted and received waveforms which consists of losses due to the forward path from the radar to the target of interest, the scattering properties of the target itself, and the path from the target back to the radar. The target channel is said to have a random and time-evolving state $h_{t} \in \mathcal{H}$, where $|\mathcal{H}| < \infty$. This state represents the scattering characteristics of the target as well as disturbances due to noise, interference, and clutter due to the forward and backward channels. The channel state, input and target parameter vector dictates the received signal via the following probability mass function
\begin{equation}
	p(y_{t}|w_{t},h_{t}) = \sum_{\mathbf{z} \in \mathcal{Z}} p(y_{t}|\mathbf{z},h_{t}) p(\mathbf{z}|w_{t},h_{t}),
\end{equation}
where $y_{t} \in \mathcal{Y}$ is the received signal\footnote{The set of all possible received signals $\mathcal{Y}$ is not necessarily finite.} at PRI $t$. Due to the unknown nature of the radar's environment, very little can be assumed about this probabilistic model in general, which emphasizes the need for an online learning approach.

The channel state is not necessarily independent of the previously transmitted waveforms\footnote{Arbitrary dependence between the transmitted waveform and channel state is reflective of cases where other systems may react to the radar's behavior.}, and evolves via a finite-memory stochastic process with transition probabilities 
\begin{equation}
p \left(h_{t+1}|\{h_{s}\}_{s=t-V}^{t},\{w_{s}\}_{s=t-V}^{t} \right),
\end{equation}
for some fixed $V < \infty$ that is unknown to the radar \emph{a priori}. The received signal is then dependent on the transmitted waveform, target parameters, and the channel state
\begin{multline}
y_{t}(t^{\prime},h_{t},\mathbf{z}_{t},w_{t}) =  \sum_{k = 1}^{N_{\texttt{tar}}} g_{k}[\mathbf{z}_{t}] w_{t}(t^{\prime}-\tau_{k}[\mathbf{z}_{t}]) \\ \exp\left(-j v_{k}[\mathbf{z}_{t}] f_{t} (t^{\prime} - \tau_{k}) / c \right) + \tilde{n}[h_{t}],
\end{multline}
where $N_{\texttt{tar}}$ is the number of targets in the scene, $g_{k}[\mathbf{z}_{t}]$, $\tau_{k}[\mathbf{z}_{t}]$, and $v_{k}[\mathbf{z}_{t}]$ are the gain, round-trip delay, and velocity due to target $k$ given target parameter vector $\mathbf{z}_{t}$, $c$ is the speed of light, and $\tilde{n}[h_{t}]$ is a random quantity which captures disturbances due to interference, clutter, and noise given channel state $h_{t}$. Due to the random nature of the noise process, it is impossible to precisely determine $y_{t}$ from the transmitted $w_{t}$ even when the target properties are known \textit{a priori}. Since we wish to avoid making strong assumptions about the channel's behavior, we note that these probabalistic relationships must be learned through repeated experience.

To aid in the waveform selection process, it is assumed that the radar has the ability to passively sense activity in the shared channel during each PRI, using the approach described in \cite{Martone2018}. The spectrum sensing process yields a vector $\mathbf{\hat{s}}_{t} = [\hat{s}_{1},...,\hat{s}_{S}]$, which contains information about the interference power in a fixed number $S \in \mathbb{N}$ of sub-channels of predetermined size. Each element of the estimated interference vector is a binary value,\footnote{The total number of unique values the interference vector can take is thus $2^{S}$} $\hat{s}_{i} \in \{0,1\}$ where zero denotes that the average interference power in band $i$ is below harmful threshold $I$ and one corresponds to average interference power above $I$ in band $s_{i}$. For example, in the case of $S = 10$ sub-channels, the estimated interference vector $\mathbf{\hat{s}}_{t} = [1,1,0,0,0,0,0,0,0,0]$ corresponds to harmful interference in only the first and second sub-channels. 

The finite set of possible interference vectors is denoted by $\mathscr{S}$. In general, it is not guaranteed that the binary interference vector measured by the spectrum sensing process precisely captures the activity in the channel. Thus, the sensed interference vector $\mathbf{\hat{s}}_{t}$ can be thought of as \textit{side information}, used to estimate the time-evolving binary interference state $\mathbf{s}_{t} \in \mathscr{S}$. The binary interference state $\mathbf{s}_{t} = [s_{1},...,s_{S}]$ is defined similar to $\mathbf{\hat{s}}_{t}$. For example $\mathbf{s}_{t} = [0,0,1,1,0,0,0,0,0,0]$ corresponds to harmful interference in only the third and fourth of $S = 10$ sub-channels. 

\subsection{Waveform Selection Problem}
Taking the perspective of \cite{Bell1993}, a reasonable strategy for waveform selection would be to find a distribution of waveforms such that the resulting received signals provide a maximum amount of information about the random target parameters the radar wishes to measure. In a general sense, this corresponds to the optimization problem
\begin{equation}
	\begin{aligned}
\max_{p(W)} I(Y;Z|W, \hat{S}) &= H(Y|W=w_{i},\hat{S}=\mathbf{\hat{s}}_{t}) \\ &-H(Y|W=w_{i},Z=\mathbf{z}_{t},\hat{S}=\mathbf{\hat{s}_{t}})\\
							  &= \sum_{k} p(y_{k}|w_{j},\mathbf{\hat{s}}_{t}) \log p(y_{k}|w_{j},\mathbf{\hat{s}}_{t})  \\ &\; \; \; + \sum_{i} \sum_{j} p(z_{j}|w_{i},\mathbf{\hat{s}}_{t})\\ & p(y_{k}|z_{i},\mathbf{\hat{s}}_{t}) \log p(y_{k}|z_{i},\mathbf{\hat{s}}_{t}) 
\end{aligned}
\end{equation}
where $I(\cdot \; ; \; \cdot)$ is the mutual information between random variables, $H(\cdot)$ is the Shannon entropy, $W$ and $Y$ are random variables describing the channel input and output, $Z$ is a random variable describing the target parameters, and $\hat{S}$ is a random variable describing information gained from spectrum sensing. However, the unknown probability distributions characterizing the target channel and lack of \textit{a priori} information about the scattering properties of the target makes this quantity hard to maximize directly. 

Additionally, it is not possible to directly estimate target range and Doppler shift from a single return alone \cite{Levanon2004}. Thus, the radar must gradually learn characteristics of the target channel via repeated experience. The learning process involves a fundamental trade-off between \emph{exploration} and \emph{exploitation} encountered in online decision problems. To formulate a sequential decision problem for radar waveform selection, we develop a cost-function which can be calculated using a combination of radar feedback from a single pulse and spectrum sensing information. This cost-based sequential optimization procedure seeks to indirectly optimize the mutual information by avoiding frequency ranges containing harmful interference while utilizing a large enough bandwidth to reduce target ambiguities. The sequential cost minimization problem can be generally formulated as follows. During each PRI, the CR wishes to select the waveform $w_{t}^{*}$ which solves the optimization problem 
\begin{equation}
w_{t}^{*} = \underset{w_{i} \in \mathcal{W}}{\argmin} \; \; \E \left[ C(w_{i},\mathbf{s}_{t})|\mathbf{\hat{s}_{t}}, \mathcal{F}_{t-1} \right], 
\end{equation}
where the expectation is taken over randomness in the target channel, $C: \mathcal{W} \times \mathscr{S} \mapsto [0,1]$ is a function which decides the relative \emph{cost} associated with transmitting $w_{i} \in \mathcal{W}$ when the true interference vector is $\mathbf{s}_{t}$, and $\mathcal{F}_{t-1}$ is the smallest $\sigma$-algebra \footnote{For the reader unfamiliar with $\sigma$-algebras, $\mathcal{F}_{t-1}$ can be simply be thought of as a set containing all the events which have occurred until time $t-1$.} generated from the history of transmitted waveforms, interference vector estimates, and observed costs until PRI $t-1$. It is assumed that the radar can store $\mathcal{F}_{t-1}$ in memory each PRI to enable information gain as the radar gains experience. Due to time-delay or possible estimation error, the true interference vector $\mathbf{s}_{t}$ is unknown when the waveform is selected. Therefore, the radar must make an informed guess as to which waveform will minimize the cost function given the interference estimate $\mathbf{\hat{s}_{t}}$ and the history $\mathcal{F}_{t-1}$.

\begin{figure*}[t]
	\centering
	\begin{subfigure}[]
		\centering
		\includegraphics[scale=0.5]{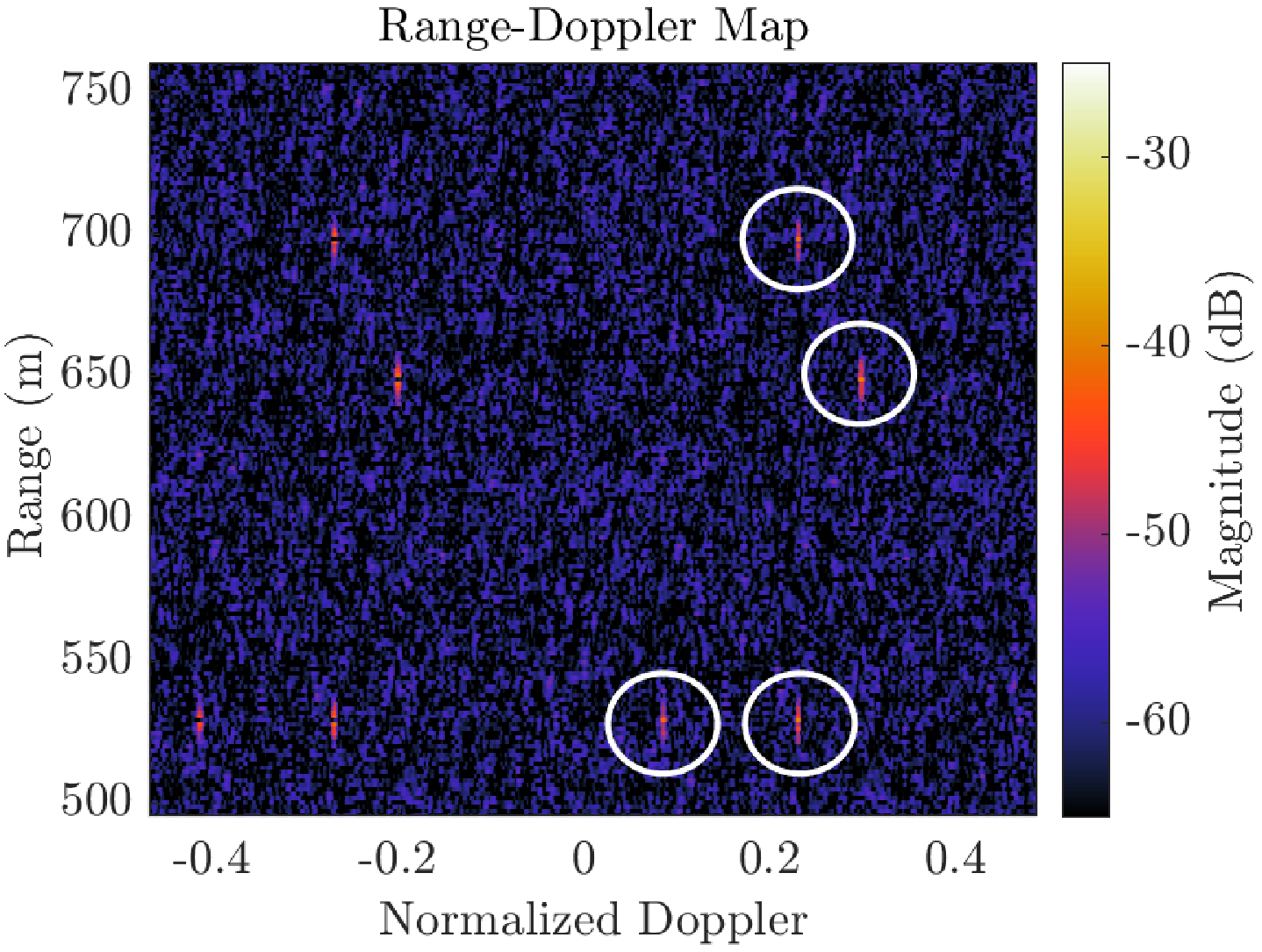}
	\end{subfigure}
	\begin{subfigure}[]
		\centering
		\includegraphics[scale=0.5]{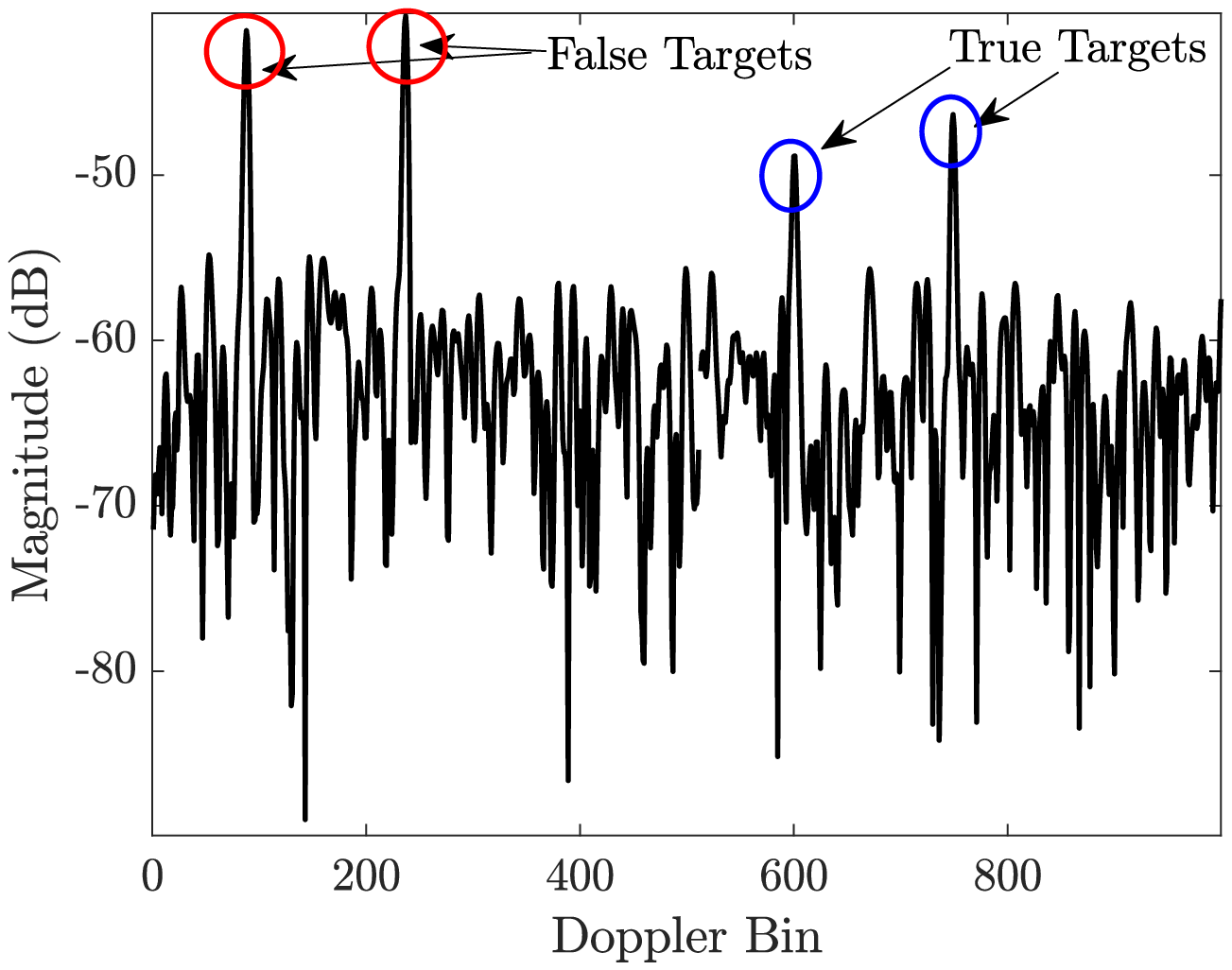}
	\end{subfigure}
	\caption{\textsc{An example of harmful distortion effects due to pulse agility.} In (a) the range-Doppler map is seen. The four true targets are located in white circles. Other features in the image are due to distortion effects and noise. In (b) the Doppler profile at a range of 530m is examined. Pulse agility results in two sharp peaks that could be interpreted by a detection algorithm as false targets.}
	\label{fig:traj}
\end{figure*}

We now proceed to a description of the cost function, which requires the following definitions.
\begin{definition}
	Define the collision bandwidth as 
	\begin{multline}
	\texttt{BW}_{c}(w_{i},\mathbf{s}_{t}) \triangleq \frac{B}{S}  \sum_{\ell = 1}^{S} \\ \left( \mathbbm{1} \left\{\left( \frac{\ell B}{S} - \frac{B}{2S} \right) \in [f_{i}-\texttt{BW}_{i}, f_{i}+\texttt{BW}_{i}] \right\} \times \mathbbm{1}\{s_{\ell} = 1 \} \right), 
	\end{multline}
	where $B$ is the bandwidth of the shared channel, $f_{i}$ is the carrier frequency of $w_{i}$ and $\texttt{BW}_{i}$ is the bandwidth of $w_{i}$. Thus, $\texttt{BW}_{c}$ corresponds to the portion of the shared channel bandwidth occupied by both the radar's waveform $w_{i}$ and the interference $\mathbf{s}_{t}$. Further, $0 \leq \texttt{BW}_{c} \leq B$.
\end{definition}
\begin{definition}
	Define the missed bandwidth as
	\begin{equation}
	\texttt{BW}_{miss}(w_{i},\mathbf{s}_{t}) \triangleq \texttt{BW}_{i^{*}} - \texttt{BW}_{i},
	\end{equation}
	where $BW_{i^{*}}$ is the bandwidth of waveform $w_{i^{*}} \in \mathcal{W}$ given by
	\begin{equation}
	\begin{aligned}
		w_{i^{*}} = \quad &\underset{w_{i} \in \mathcal{W}}{\argmax} \quad \texttt{BW}_{i}\\
		&\text{subject to} \quad \texttt{BW}_{c}(w_{i},\mathbf{s}_{t}) = 0
	\end{aligned}
	\end{equation}	
	 If there is no $w_{i}$ in $\mathcal{W}$ which has zero collision bandwidth with $\mathbf{s}_{t}$, then $\texttt{BW}_{miss} = 0$. Thus, $0 \leq \texttt{BW}_{miss} \leq B$.
\end{definition}
\begin{definition}
	Let the waveform cost function be\footnote{The notation $C_{t} = C(w_{t},\mathbf{s}_{t})$ is also used for brevity.}
	\begin{equation}
	C(w_{t},\mathbf{s}_{t}) \triangleq \beta_{1} \texttt{BW}_{c}(w_{t},\mathbf{s}_{t}) + \beta_{2} \texttt{BW}_{miss}(w_{t},\mathbf{s}_{t}) 
	\label{eq:cost}
	\end{equation}
	where the parameters $0 \leq \beta_{1},\beta_{2} \leq \frac{1}{B}$ are selected such that $C(w_{i},\mathbf{s}_{t})$ is bounded in $[0,1]$. The weighting of these parameters define the radar's operation preference. To calculate the cost, it is assumed the radar can recover the interference state vector $\mathbf{s}_{t}$ from the most recent received signal $y_{t}$ via \texttt{SINR} estimation.
\end{definition}
\begin{remark}
	The cost function (\ref{eq:cost}) is used as a surrogate measure for received target information in a single pulse. Minimization of the collision bandwidth ensures that the targets of interest can be reliably detected in an interference-limited scenario (ie. the interference-to-noise ratio is larger than the signal-to-noise ratio). Minimization of the missed bandwidth ensures sufficient range resolution for estimation, as dictated by the Cramér-Rao lower bound for range estimation \cite{CRLB}. 
\end{remark}
\begin{prop}
	The CR cost function $C(w_{i},\mathbf{s}_{t})$ is locally Lipschitz continuous in the first argument, meaning that for any two waveforms $w_{j} \in \mathcal{W}$ and $w_{k} \in \mathcal{W}$, $\lvert C(w_{j}) - C(w_{k}) \rvert \leq \mathcal{L}(w_{j},w_{k})$, where $\mathcal{L}$ is a metric of distance between waveforms.
\end{prop}
\begin{proof} \renewcommand{\qedsymbol}{}
	See Appendix A.
\end{proof}
\begin{remark}
The waveform selection set forms a \emph{metric space} $(\mathcal{W},\mathcal{L})$, which allows for similarity between waveforms in terms of metric $\mathcal{L}$ to be exploited by online learning algorithms. This is beneficial for cases where $|\mathcal{W}|$ is large, since the radar can learn information about several waveforms from a single transmission.
\end{remark}

By selecting waveforms which yield low average cost in terms of (\ref{eq:cost}), the CR is equivalently attempting to minimize the cumulative \textit{strong regret} experienced in period $n$, defined by
\begin{equation}
\texttt{Regret}(n) \triangleq \sum_{t = 1}^{n} \left[ C({w_{t}},\mathbf{s}_{t}) - C({w_{t}^{*}},\mathbf{s}_{t})  \right],
\label{eq:regret} 
\end{equation}
where $w_{t}^{*}$ is the waveform which minimizes $C(\cdot,\mathbf{s}_{t})$ in PRI $t$ and $w_{t}$ is the waveform transmitted by the CR in PRI $t$. Since calculation of (\ref{eq:regret}) requires knowledge of $w_{t}^{*}$ at each step, it must be calculated in hindsight. Thus, (\ref{eq:regret}) cannot be minimized directly, and online optimization is used to select waveforms in each PRI such that $w_{t}^{*}$ is selected by the radar as often as possible in expectation.

Since a variety of waveforms may be transmitted within a CPI, we also wish to mitigate distortion effects in the range-Doppler image due to high sidelobe levels as a result of pulse-diversity (discussed in Chapter 9 of \cite{Levanon2004}). An example of undesirable effects in the range-Doppler image can be seen in Figure \ref{fig:traj}. To achieve a desirable level of discrimination in both delay and Doppler dimensions, we adopt the viewpoint of the radar as an imaging system \cite{Guey1998}, and attempt to limit blurring in the point-spread function by constraining the set of available waveforms at each time step. This is performed by limiting the distance between adjacent transmitted waveforms using the following distortion function as a metric.

\begin{definition}
	Let the distortion function be
	\begin{equation}
		D(w_{t},w_{t-1}) \triangleq \gamma_{1} \lVert f_{t} - f_{t-1} \rVert^{2} + \; \gamma_{2} \lVert \texttt{BW}_{t} - \texttt{BW}_{t-1} \rVert^{2} 
		\label{eq:disto}
	\end{equation}
	where $f_{t}, f_{t-1}$ and $\texttt{BW}_{t}, \texttt{BW}_{t-1}$ are the respective center frequencies and bandwidths of $w_{t}$ and $w_{t-1}$. The parameters $\gamma_{1}$ and $\gamma_{2}$ are selected such that $D(w_{t},w_{t-1}) \in [0,1]$ for all $w_{t} \in \mathcal{W}$. The distortion function can be used as a distance metric to define the metric space $(\mathcal{W},D)$.
\end{definition}

Then the radar can solve the following constrained optimization problem at each step
\begin{equation}
	\begin{aligned}
		& \underset{w_{i} \in \mathcal{W}}{\text{minimize}}
		& & \E[C(w_{i},\mathbf{s}_{t})| \mathbf{\hat{s}}_{t}, \mathcal{F}_{t-1}] \\
		& \text{subject to}
		& & D(w_{i},w_{t-1}) < \hat{d},
	\end{aligned}
	\label{eq:const}
\end{equation}
where $\E[C({w_{i}},\mathbf{s}_{t})|\mathbf{\hat{s}}_{t}, \mathcal{F}_{t-1})]$ is the posterior expected cost of transmitting $w_{i}$ based on the information observed until the previous PRI and $\hat{d}$ is the distance corresponding to a tolerable level of distortion. This time-varying constraint is implemented by assembling a subset of $\mathcal{W}$ during each PRI. This limited waveform catalog, defined at each PRI to be
\begin{equation}
 \mathcal{W'} \triangleq \{w_{i} \in \mathcal{W}: D(w_{i},w_{t-1}) < \hat{d} \},
\end{equation} 
 ensures that the constraint in (\ref{eq:const}) is met. Further, this time-varying waveform catalog can be utilized directly in the linear contextual bandit learning framework described in the next section.

\section{Online Learning Framework}
\label{se:online}
To select the waveform which is optimal in terms of expected cost during each PRI, the radar must use past experience to estimate $\E[C(w_{i},\mathbf{s}_{t})|\mathbf{\hat{s}}_{t}, \mathcal{F}_{t-1}]$ for each unique waveform and interference state pair $(w_{i}, \mathbf{\hat{s}}_{t}) \in \mathcal{W} \times \mathscr{S}$. Thus, a natural trade-off between \textit{exploration} and \textit{exploitation} arises. Each waveform must be transmitted a sufficient number of times across different interference contexts such that the expected cost can be reliably predicted using only the information in $\mathbf{\hat{s}_{t}}$ and $\mathcal{F}_{t-1}$. Additionally, the total number of sub-optimal waveforms transmitted in period $n$ should be minimized. To balance exploration and exploitation, the problem is formulated using both stochastic and adversarial linear contextual bandit models. These models generalize the multi-armed bandit problem by allowing the decision-maker to utilize side information when selecting actions. To address concerns regarding distortion effects, both schemes utilize a time-varying action set that limits actions which may lead to distortion effects in the range-Doppler map as defined in (\ref{eq:disto}).

\subsection{Stochastic Linear Contextual Bandits and Thompson Sampling}
We first examine a stochastic linear contextual bandit learning model, under which the cost at each PRI is characterized by the following structure
\begin{equation}
	C(w_{i}, \mathbf{s}_{t}) = \langle \boldsymbol{\theta}, \mathbf{x}_{w_{i},t} \rangle + \eta_{t},
	\label{eq:stochmodel}
\end{equation}
where $\boldsymbol{\theta} \in \mathbb{R}^{d}$ is a parameter vector that the radar wishes to learn, $\mathbf{x}_{w_{i},t} \in \mathbb{R}^{d}$ is a \emph{context} vector associated with each waveform $w_{i}$ at time $t$. The context vector is assembled using information about previous observations, transmitted waveforms and costs from the $\sigma$-algebra $\mathcal{F}_{t-1}$, which is stored in memory. In this implementation, the context vector contains the following features, 
\begin{equation}
	\begin{aligned}
		\xi_{1} = \bar{C}(w_{i},\mathbf{\hat{s}}_{t}), \; \; & \xi_{2} = \frac{\sum_{\ell}(C_{\ell}(w_{i},\mathbf{\hat{s}}_{t})-\bar{C}(w_{i},\mathbf{\hat{s}}_{t}))^{2}}{N_{c}-1},\\ & \text{and} \; \; \xi_{3} = C_{N_{c}}(w_{i},\mathbf{\hat{s}}_{t}),
	\end{aligned}
\end{equation}
where $N_{c}$ is the number of times the context-action pair $(w_{i},\mathbf{\hat{s}}_{t})$ has been encountered. $\xi_{1}$ is the sample mean of all observed $C(w_{i},\mathbf{\hat{s}}_{t})$ instances where the sum is taken over a sub $\sigma$-algebra of $\mathcal{F}_{t-1}$ that contains only instances of the context-action pair of interest. $\xi_{2}$ is the sample variance of costs over the same sub $\sigma$-algebra containing instances of $C(w_{i},\mathbf{\hat{s}}_{t})$. $\xi_{3}$ is the most recently observed instance of $C(w_{i},\mathbf{\hat{s}}_{t})$, which gives near-term information. 

Returning to the description of the model (\ref{eq:stochmodel}), $\eta_{t}$ is a random disturbance, which reflects cases in which the inner product $\langle \boldsymbol{\theta}, \mathbf{x}_{w_{i},t} \rangle$ does not explicitly predict the cost, possibly due to fluctuations in the environment or estimation errors. The linear relationship between the context vectors and the costs through an inner product with $\boldsymbol{\theta}$ allows for learning to transfer between contexts, which is a powerful tool when particular contexts may occur infrequently.

It is assumed that the distribution of $\eta_{t}$ is conditionally 1-subgaussian \cite{Lattimore2020}, which precisely means that for any $\lambda \in \mathbb{R}$
\begin{equation}
	\E[\exp(\lambda \eta_{t}) |\mathcal{F}_{t}] \leq \exp\left(\frac{\lambda^{2}}{2}\right) \; \; \; \textit{almost surely},
\end{equation}
which implies that $\eta_{t}$ has a tail that decays faster than a Gaussian distribution. Practically, this means that there exists a parameter vector $\boldsymbol{\theta}$ such that knowledge of $\boldsymbol{\theta}$ will allow the radar to select the waveform with the lowest expected cost very often. In this setting, the regret can then be expressed as 
\begin{multline}
\texttt{Regret}(n)= \mathbb{E} [\textstyle \sum_{t=1}^{n} C_{t}(w_{t},\mathbf{s}_{t}) \\ - \textstyle \sum_{t=1}^{n} \min _{w_{i} \in \mathcal{W'}} \langle \boldsymbol{\theta}, \mathbf{x}_{w_{i}} \rangle ].
\end{multline}

Previous work in radar and communications has found the stochastic model to be viable for many wireless transmission problems due to the underlying randomness of the channel conditions \cite{dsa1,Thornton2020,Amuru2016}. Many efficient algorithms have been well-studied in the stochastic setting, such as upper confidence bound and $\epsilon$-greedy strategies \cite{Lattimore2020}. However, a Bayesian inspired heuristic called Thompson Sampling (TS) has attracted significant attention in the online learning literature due to near-optimal empirical performance on a variety of tasks, which has recently been by supplemented by theoretical results offering favorable performance guarantees as well as computationally efficient implementations \cite{agrawal}, \cite{infoTS}.
\begin{algorithm}[t]
	\setlength{\textfloatsep}{0pt}
	\caption{Constrained Linear Contextual Thompson Sampling Waveform Selection}
	\SetAlgoLined
	Initialize parameters $\mathbf{B}_{1} = \mathbf{I}_{d}$, $\hat{\boldsymbol{\theta}} = \mathbf{0}_{d}$, $\mathbf{f} = \mathbf{0}_{d}$, $\hat{d}$;\\
	\For{t = 2, ..., $n$}{
		\vspace{0.1cm}
		Sense spectrum to estimate interference vector $\mathbf{\hat{s}}_{t} = [s_{1},...s_{S}]$;\vspace{0.2cm}
		
		Create constrained action space $\mathcal{W'} = \{ w_{i} \in \mathcal{W}: D(w_{i}|w_{t-1}) < \hat{d} \}$; \vspace{0.2cm}
		
		Using $\mathbf{\hat{s}}_{t}$ and $\mathcal{F}_{t-1}$ assemble context vectors $\mathbf{x}_{w_{i},t} = [\xi_{1},...,\xi_{d}], \; \; \forall \; w_{i} \in \mathcal{W'}$;\vspace{0.2cm}
		
		Sample $\tilde{\boldsymbol{\theta}} \sim \mathcal{N}(\hat{\boldsymbol{\theta}}, \mathbf{B}_{t}^{-1})$;\vspace{0.2cm}
		
		Select LFM waveform $\mathbf{w}_{i}(t) = \underset{w_{i} \in \mathcal{W}}{\argmin} \langle \mathbf{x}_{w_{i},t}, \tilde{\boldsymbol{\theta}} \rangle$;\vspace{0.2cm}
		
		Observe cost $C(w_{t}, \mathbf{s}_{t})$; \vspace{0.2cm}
		
		Update distribution parameters $\mathbf{B}_{t} = \mathbf{B}_{t} + \mathbf{x}_{w_{i},t} \mathbf{x}_{w_{i},t}^{T}$, $\mathbf{f} = \mathbf{f} + \mathbf{x}_{w_{i},t} C_{t}$, and $\boldsymbol{\hat{\theta}} =\mathbf{B}_{t}^{-1} \mathbf{f}$;
	}
    \label{algo:tsamp}
\end{algorithm}	

TS simply involves selecting actions based on the posterior probability of being cost-optimal. The posterior distribution $\mathbb{P}(\boldsymbol{\theta}|\mathcal{F}_{t-1})$ is computed using Bayes' rule and a randomly initialized normal prior\footnote{The implementation in this paper considers Gaussian Thompson Sampling, which exploits the normal-normal conjugacy property to yield a normally distributed posterior from which samples can be efficiently generated.}, which can be easily updated. Given context vector $\mathbf{x}_{w_{i}}$ and parameter vector $\boldsymbol{\theta}$, the likelihood of receiving cost $C(w_{i},\mathbf{s}_{t})$ is
\begin{equation}
	L(C(w_{i},\mathbf{s}_{t})| \boldsymbol{\theta}) \sim \mathcal{N}(\mathbf{x}_{w_{i}}^{T} \boldsymbol{\theta}, v^{2}),
\end{equation}
where $v$ is an exploration parameter that specifies the algorithm. We can then place a Gaussian prior distribution on $\boldsymbol{\theta}$ given by
\begin{equation}
	\mathbb{P}(\boldsymbol{\theta}) \sim \mathcal{N} (\hat{\boldsymbol{\theta}}, v^{2} \mathbf{B}_{t}^{-1}).
\end{equation}

Applying Bayes' rule, the posterior distribution on $\boldsymbol{\theta}$ can then be written up to a constant factor as
\begin{equation}
	\begin{split}
		\mathbb{P}(\tilde{\boldsymbol{\theta}} | C(w_{i},\mathbf{s}_{t})) & \propto L(C(w_{i},\mathbf{s}_{t}) | \boldsymbol{\theta}) \mathbb{P}(\boldsymbol{\theta}) \\
		& \propto \mathcal{N}(\hat{\boldsymbol{\theta}}, v^{2}\mathbf{B}_{t}^{-1}),
	\end{split}
\end{equation}
where the posterior mean and covariance matrix can be expressed as
\begin{align}
		\mathbf{B}_{t} &= \mathbf{I}_{d} + \textstyle \sum_{\tau = 1}^{t-1} \mathbf{x}_{\mathbf{a}(\tau)}(\tau) \mathbf{x}_{\mathbf{a}(\tau)}^{T}(\tau), \\
		\hat{\boldsymbol{\theta}} &= \mathbf{B}_{t}^{-1} \textstyle \sum_{\tau=1}^{t-1} \mathbf{x}_{\mathbf{a}(\tau)}(\tau) r_{\mathbf{a}(\tau)}(\tau),
	\end{align}
where $\mathbf{I}_{d}$ is the $d$-dimensional identity matrix. Thus, the posterior estimate of the model $\tilde{\boldsymbol{\theta}}$ can be efficiently sampled from a $d$-dimensional multivariate normal distribution. The distribution parameters $\hat{\boldsymbol{\theta}}_{t}$ and $\mathbf{B}_{t}^{-1}$ can also be easily updated based on the context and reward received at time $t$.

The full procedure of building the posterior distribution and obtaining samples can be seen in Algorithm \ref{algo:tsamp}. We note that the TS algorithm is computationally efficient as long as (\ref{eq:const}) is easily solvable. This is the case when the inner product $\langle \boldsymbol{\theta}, \mathbf{x}_{w_{i},t} \rangle$ can be computed and a sample can be efficiently generated from $\mathbb{P}(\boldsymbol{\theta}|\mathcal{F}_{t-1})$, which is expected to hold for most waveform selection problems of practical interest.

The algorithm is only computationally limited by the context dimension $d$, as the only real-time tasks are sampling from a multivariate Gaussian distribution of dimension $d$ and updating the associated distribution parameters. Since the radar's waveform catalog is of finite cardinality $\lvert \mathcal{W} \rvert = W < \infty$, the action set is further constrained such that $\lvert \mathcal{W'} \rvert \leq \lvert \mathcal{W} \rvert$ each PRI, and the context dimensionality can be kept low for most waveform selection problems, the TS algorithm is very efficient in most practical cases. This is in contrast to other stochastic linear bandit algorithms with similar regret guarantees, such as linUCB, where the maximization step is often intractable for large action spaces \cite{Lattimore2020}.

\begin{prop}
	The expected (Bayesian) regret of Algorithm \ref{algo:tsamp} under the time-varying waveform catalog $\mathcal{W'}$ and 1-subgaussian disturbance assumption is 
	\begin{equation}
	   \E[\texttt{Regret}(n)] = \mathcal{O}(3\sqrt{n log(n)}),
	   \label{eq:bayesTS} 
	 \end{equation}
 	while the frequentist (worst-case) regret is given by
 	\begin{equation}
 		\texttt{Regret}(n) = \tilde{\mathcal{O}}(\sqrt{27n}).
 		\label{eq:freqTS}
 	\end{equation}

	Note that neither bound is dependent on the size of the waveform catalog. However, a large waveform catalog can increase the uncertainty in the prior distribution, which impacts the distribution of the regret. For a detailed description on the role of the prior, see the analysis of \cite{infoTS}.
\end{prop}

\begin{proof}\renewcommand{\qedsymbol}{}
	See Appendix B.
\end{proof}

\subsection{Adversarial Linear Contextual Bandits and the EXP3 Algorithm}
In addition to the stochastic linear contextual bandit model, we also consider the \textit{adversarial} linear contextual bandit problem. This formulation removes the assumption that the costs for each context waveform pair are sampled from a stationary distribution, and makes the gentler assumption that costs are arbitrarily selected by an intelligent adversary \cite{Lattimore2020}. This is a more general setting than the stochastic bandit which relaxes the assumption that costs are drawn from a fixed model $\langle \boldsymbol{\theta}, \mathbf{x}_{w_{i},t} \rangle + \eta_{t}$. The cost structure in the adversarial linear contextual bandit model is expressed by
\begin{equation}
	C(w_{i}, \mathbf{s}_{t}) = \langle \boldsymbol{\theta}_{t}, \mathbf{x}_{w_{i},t} \rangle,
\end{equation}
where the parameter vector that the radar wishes to learn $\boldsymbol{\theta}_{t} \in \mathbb{R}^{d}$ is now a time-varying quantity. Thus, radar must learn a model which is nonstationary in general. The only source of randomness in the radar's regret is the distribution of the waveforms the radar transmits, which may reflect scenarios when the radar's waveform impacts the channel quality, such as in the presence of an intelligent jammer.

\begin{algorithm}[t]
	\setlength{\textfloatsep}{0pt}
	\caption{Constrained Linear Contextual EXP3 Waveform Selection}
	\SetAlgoLined
	Initialize learning rate $\varepsilon \in (0,1)$, exploration distribution $\pi$, tolerable distortion level $\hat{d}$, and exploration parameter $\gamma \in [0,1]$\\
	\For{t = 2,...,$n$}{
		\vspace{0.1cm}
		Sense spectrum and estimate interference vector $\mathbf{\hat{s}}_{t} = [s_{1},...s_{S}]$;\vspace{0.2cm}		
		
		Create constrained action space $\mathcal{W'} = \{ w_{i}  \in \mathcal{W}: D(w_{i}|w_{t-1}) < \hat{d} \}$; \vspace{0.2cm}
		
		Using $\mathbf{\hat{s}}_{t}$ and $\mathcal{F}_{t-1}$ assemble context vectors $\mathbf{x}_{w_{i},t} = [\xi_{1},...,\xi_{d}], \; \; \forall \; w_{i} \in \mathcal{W'}$;\vspace{0.2cm}
		
		For each $w_{i} \in \mathcal{W'}$ \vspace{.1cm} set $P_{t}(w_{i}) \leftarrow \gamma \pi(w_{i})+(1-\gamma) \frac{\exp \left(-\varepsilon \sum_{j=1}^{t-1} \hat{C}_{j}(w_{i}, \mathbf{s}_{j})\right)}{\sum_{w_{i}^{\prime} \in \mathcal{W'}} \exp \left(-\varepsilon \sum_{j=1}^{t-1} \hat{C}_{j}\left(w_{i}^{\prime}, \mathbf{s}_{j} \right)\right)}$; \vspace{0.2cm}
		
		Sample $w_{t} \sim P_{t}$ and observe $C(w_{t},\mathbf{s}_{t})$; \vspace{.2cm}
		
		Set $\boldsymbol{\hat{\theta}}_{t} \leftarrow \mathbf{Q}_{t}^{-1} \mathbf{x}_{w_{i},t} C_{t}$ and $\hat{C}_{t}(w_{i},\mathbf{s}_{t}) \leftarrow \langle \mathbf{x}_{w_{i},t}, \boldsymbol{\hat{\theta}}_{t} \rangle$;						
	}
	\label{algo:exp3algo}	
\end{algorithm}

To balance exploration and exploitation in this setting, the radar uses a constrained variant of the EXP3 algorithm, first introduced by Auer in \cite{Auer2002} and widely used in adversarial bandit problems thereafter. EXP3 uses \textit{exponentially weighted estimation} to approximate the expected cost of each waveform before transmission. The exponentially weighted estimator $P_{t}: \mathcal{W'} \mapsto [0,1]$ is given by the probability mass function
\begin{equation}
	\tilde{P}_{t}(w_{i}) \propto \exp \left(\varepsilon \sum_{j=1}^{t-1} \hat{C}_{j}(w_{i},\mathbf{s}_{j})\right)
	\label{eq:expweight}
\end{equation}
where $\varepsilon \in (0,1)$ is the learning rate and $\hat{C}_{j}$ is an estimate of the cost at PRI $j$. To control the variance of the cost estimates, the probability mass in (\ref{eq:expweight}) is mixed with an arbitrary exploration distribution $\pi: \mathcal{W} \mapsto [0,1]$ where $\sum_{w_{i} \in \mathcal{W'}} \pi(w_{i}) = 1$. In the constrained implementation, $\pi$ is set to be uniform over $\mathcal{W'}$ each PRI. The mixture distribution from which waveforms are selected is then given by
\begin{equation}
	P_{t}(w_{i})=(1-\gamma) \tilde{P}_{t}(w_{i})+\gamma \pi(w_{i}),
\end{equation}
where $\gamma \in [0,1]$ is a mixing factor. Each PRI, the waveform is then sampled $w_{t} \sim P_{t}$. However, calculation of $P_{t}$ involves an estimation of the cost for each $w_{i} \in \mathcal{W'}$ at each time step. This is performed via least-squares. An estimate of the model is obtained using 
\begin{equation}
	\boldsymbol{\hat{\theta}}_{t} = \mathbf{Q}_{t}^{-1} \mathbf{x}_{w_{i},t} C_{t},
	\label{eq:lsq} 
\end{equation}
where 
\begin{equation}
\mathbf{Q}_{t} = \sum_{w_{i} \in \mathcal{W'}} P_{t}(w_{i}) \mathbf{x}_{w_{i},t} \mathbf{x}_{w_{i},t}^{T}, 
\end{equation}
is a nonsingular matrix in $\mathbb{R}^{d \times d}$. This can be ensured by selecting the exploration distribution such that the matrix
\begin{equation}
Q(\pi) = \sum_{w_{i} \in \mathcal{W'}} \pi(w_{i}) \mathbf{x}_{w_{i},t} \mathbf{x}_{w_{i},t}^{T}, 
\end{equation}
is nonsingular. Once (\ref{eq:lsq}) is computed, the cost estimate is easily accessed through the inner product $\langle \mathbf{x}_{w_{i},t}, \boldsymbol{\hat{\theta}}_{t} \rangle$. 

EXP3 is nearly optimal in terms of worst-case regret, but the distribution of costs has a high variance \cite{Lattimore2020}. However, through the bias-variance trade-off, this general model allows a CR to maintain acceptable performance in a wide range of environments. A description of the constrained EXP3 waveform selection approach used in this work can be seen in Algorithm \ref{algo:exp3algo}.

The computational complexity of the proposed EXP3 algorithm is $\mathcal{O}(Wd+d^{3})$ per round \cite{Lattimore2020}, where $W$ is size of the waveform catalog and $d$ is the dimensionality of the context vector. For most waveform selection problems, the dimensionality can be kept reasonably low (our examples use $d=3$), and should not be a practical issue. However, problems may arise when the number of waveforms $W$ is large. Solving adversarial bandit problems with many arms is difficult in general, and is the subject of much current investigation in the online learning theory literature.

\begin{prop}
	The worst-case regret of the constrained EXP3 algorithm for $W = |\mathcal{W}|$ is given by
	\begin{equation}
		\texttt{Regret}(n) = \mathcal{O}(2\sqrt{9 n \log{W}}),
		\label{eq:exp3bound}
	\end{equation}
    where, unlike TS, the bound is sublinearly dependent on the size of the waveform catalog.
\end{prop}

\begin{proof}\renewcommand{\qedsymbol}{}
	See Appendix C.
\end{proof}

\section{Considerations for Single Target Tracking}
\label{se:track}
Now that the waveform selection process has been described, and an online learning framework to select optimal waveforms has been presented, we discuss how the proposed approach can be applied to a target tracking system. In legacy radars, the tracking system or \emph{processor}, has been considered separately from the selection of sensor parameters. However, recent works have attempted to unify optimization of the sensor and processor sub-systems for improved system-level performance \cite{framework}. Since the online learning approach presented here considers side information, and the belief of the algorithm can be easily updated in the presence of new data, information from the tracking system can be used to guide the sensor optimization problem and vice versa.

Let the target's position be given by the state vector $\nbx_{k} \in \nbbR^{n_{x}}$, where $n_{x}$ is the dimension of target information considered and $k \in \nbbN$ is the tracker's time index, which consists of many radar PRIs. The target's state is assumed to evolve according to a Markov process, meaning $\P(\nbx_{k}|\nbX_{k-1}) = \P(\nbx_{k}|\nbx_{k-1})$, where $\nbX_{k-1} \triangleq \{\nbx_{i}: i = 1,...,k-1 \}$. This can be expressed using the stochastic model 
\begin{equation}
	\nbx_{k} = \nbf_{k-1}(\nbx_{k-1},\nbv_{k-1}),
\end{equation}
where $\nbf_{k-1}$ is a known function and $\nbv_{k-1}$ is the process noise with a known density. The goal of the tracker is to estimate $\nbx_{k}$ using a sequence of measurements $\nbZ_{k} \triangleq \{\nbz_{i}: i = 1,...,k\}$. These measurements are derived from the target state via the model
\begin{equation}
	\nbz_{k} = \nbh_{k}(\nbx_{k},\nbv_{k}),
\end{equation}
\begin{table*}[]
	\centering
	\caption{\textsc{Simulation Parameters Considered}}
	\label{tab:sims}
	\begin{tabular}{|l|l|l|l|}
		\hline
		Parameter & Value & Parameter & Value \\ \hline
		$N_{\texttt{BS}}$ (coexistence scenario) & 90 & $\texttt{JNR}$ (jamming scenario) & $20$ $\texttt{dB}$ \\ \hline
		Intf. Bandwidth & $20 \texttt{MHz}$ & Shared channel bandwidth  & $100$ \texttt{MHz} \\ \hline
		$P_{j}$ (coexistence scenario) & $40-46.5 \texttt{dBm}$ & CPI length $M_{\texttt{CPI}}$  & $400$ pulses  \\ \hline
		$\mathbf{d}_{j}$ (coexistence) & $5$-$6$ Km & Path loss $\psi$ (coexistence) & 3.5  \\ \hline
		Distortion constraint $\hat{d}$ & $0.2$ & Pulse Repetition Interval  & $409.6 \mu s$ \\ \hline
	\end{tabular}
\end{table*}
where $\nbh_{k}$ is a known function and $\nbv_{k}$ is the measurement noise with known density. To estimate the probability density $p(\nbx_{k}|\nbZ_{k})$, a recursive sequence of \emph{prediction} and \emph{updating} can be used. \cite{Ristic2004}. The prediction step involves applying the well-known Chapman-Kolmogorov equation
\begin{equation}
	p(\nbx_{k}|\nbZ_{k-1}) = \int p(\nbx_{k}|\nbx_{k-1})p(\nbx_{k-1}|\nbZ_{k-1}) d\nbx_{k-1}.
	\label{eq:prediction}
\end{equation}

Once the next measurement $\nbz_{k}$ is received, the posterior density $p(\nbx_{k}|\nbZ_{k})$ can be found by applying Bayes' rule
\begin{align}
	p(\nbx_{k}|\nbZ_{k}) &= p(\nbx_{k}|\nbz_{k}, \nbZ_{k-1}) \\
	&= \frac{p(\nbz_{k}|\nbx_{k},\nbZ_{k-1})p(\nbx_{k}|\nbZ_{k-1})}{p(\nbz_{k}|\nbZ_{k-1})}	\\
	&= \frac{p(\nbz_{k}|\nbx_{k})p(\nbx_{k}|\nbZ_{k-1})}{p(\nbz_{k}|\nbZ_{k-1})},
	\label{eq:update}
\end{align}
where the normalizing density $p(\nbz_{k}|\nbZ_{k-1})$ is defined by $\nbh_{k}$ and $\nbv_{k}$. In general, this recursive procedure cannot be explicitly performed since an entire density must be stored, which is equivalent to an infinite vector \cite{Ristic2004}. Thus, for tractability purposes we assume the posterior $p(\nbx_{k}|\nbZ_{k})$ is normally distributed, that the noise processes $\nbv_{k-1}$ and $\nbv_{k}$ are normally distributed, and that the models $\nbf_{k-1}$ and $\nbh_{k}$ are linear. This structure implies that the Kalman filter is optimal, and can be used for the tracking procedure here. However, the general structure which follows can be applied to any particle filtering algorithm using the recursion of (\ref{eq:prediction}) and (\ref{eq:update}).

The adaptive waveform selection process effects the observation noise covariance matrix via the pulse time $T_{k}$ and sweep rate $\alpha_{k}$ at tracking interval $k$ as follows:
\begin{multline}
	N(T_{k+1},\alpha_{k+1}) = \\ \begin{bmatrix}
		c^2 T_{k+1}^{2}/(2\eta) & -c^{2} \alpha_{k+1} T_{k+1}^{2}/(f_{c}\eta)\\
		-c^{2} \alpha_{k+1} T_{k+1}^{2}/(f_{c}\eta) & \frac{c^{2}}{f_{c}\eta}(\frac{1}{2 T_{k+1}}+2 \alpha_{k+1}^{2} T_{k+1}^{2}),
	\end{bmatrix}
\end{multline}
where $\eta$ is the signal-to-noise ratio at the radar, $c$ is the speed of light, and $f_{c}$ is the carrier frequency. It can then be shown (using second derivatives as in \cite{Kershaw1994}), that the optimal parameters are
\begin{align}
	\alpha^{*}_{k+1} &= \frac{-w_{c}p_{12}}{2p_{11}}\\
	T^{*}_{k+1} &= \left( \frac{p_{11}^{2}}{\omega_{c}^{2}(p_{11}p_{22}-p_{12}^{2})} \right)^{1/4},
\end{align}
where $p_{ij}$ is the $ij^{th}$ element of the smoothed tracking error covariance matrix. Thus, the track-optimal pulse length $\alpha^{*}$ can be used directly in the transmit waveform written in (\ref{eq:tx}). Implementation of the optimal sweep rate requires slightly more care, as it impacts the radar's ability to mitigate interference. One option is to use $b^{*}_{k+1}$ directly in (\ref{eq:tx}), reducing the size of the waveform catalog. Another option is to introduce a penalty to the cost function, where large deviations from $b^{*}_{k+1}$ are penalized.

\begin{figure*}[t]
	\centering
	\begin{subfigure}[]
		\centering
		\includegraphics[scale=0.35]{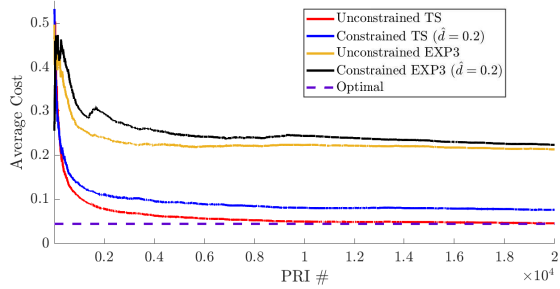}
	\end{subfigure}
	\begin{subfigure}[]
		\centering
		\includegraphics[scale=0.35]{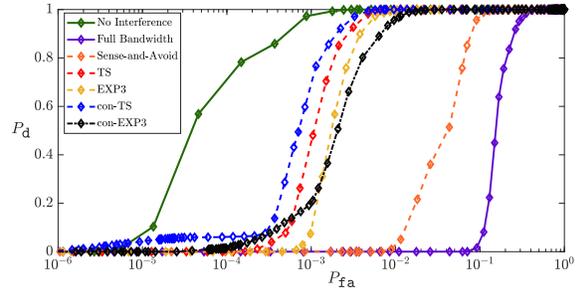}
	\end{subfigure}
	\caption{(a) \textsc{Average cost} of online learning algorithms over $n=20,000$ PRIs in the radar-communications coexistence scenario. (b) \textsc{Average detection characteristics} of each online learning algorithm compared to a static radar which occupies the entire shared channel. Detection characteristics for the online learning algorithms include the time before convergence to a stationary solution.}
	\label{fig:tc7Detect}
\end{figure*}

\section{Simulation Study}
\label{se:sim}
In this section, the proposed constrained online learning framework is evaluated in a radar-communications coexistence setting as well as in the presence of an adaptive jammer. In the former setting, the CR is the secondary user of a shared spectrum channel. The CR wishes to maximize its own detection performance by selecting waveforms which mitigate both interference and distortion effects in the processed data, while utilizing sufficient available bandwidth and causing minimal harmful interference to other systems. The CR shares the channel with $N_{\texttt{BS}}$ cellular base stations (BSs), which are spatially distant from the radar. In the latter setting, the CR also wishes to maximize detection performance while mitigating distortion effects, but a single frequency-agile jammer is capable of tracking the radar's transmitted waveforms. Each setting is further described below and associated simulation parameters are provided in Table \ref{tab:sims}.

\subsection{Radar-Communications Coexistence Scenario}
Each PRI, the radar selects a waveform $w_{i} \in \mathcal{W'}$ and observes $\mathcal{C}(w_{i},\mathbf{s}_{t})$. Once $M_{\texttt{CPI}}$ pulses are received, windowing, matched filtering and a 2D FFT are performed to create a range-Doppler map. The considered CPI's are non-overlapping. Detection analysis is performed once each CPI by applying a threshold selected by the 2D cell-averaging Constant False-Alarm Rate (CFAR) algorithm to evaluate the target detection properties of each CR scheme and traditional fixed band radar operation. Since the BSs are located far from the radar, small scale fading effects are assumed to be absent and the interference channel is dominated by correlated shadowing. The aggregate interference in each sub-channel due to the BSs at the radar is given by
\begin{equation}
	\mathcal{I}_{\texttt{agg}} =  \textstyle \sum_{j = 1}^{N_{\texttt{act}}} P_{j} \mathcal{G}_{r} \left\lVert \mathbf{d}_{j} \right\rVert^{- \psi} \exp({X_{j}}),      
\end{equation}
where $N_{\texttt{act}}$ is the number of active BSs, $P_{j}$ is the transmission power of BS $j$, $\mathcal{G}_{r}$ is the radar recieve antenna gain, $\mathbf{d}_{j}$ is the distance from BS $j$ to the radar, $\psi$ is the path loss exponent, and $X_{j} \sim N(\hat{\mu}_{j}, \sigma^{2}_{j})$ is the data transmitted by BS $j$. The cellular network bandwidth is $20 \texttt{MHz}$ and BSs transmit between $40$ and $46.5 \texttt{dBm}$. The BSs are randomly distributed between $5$ and $6 \texttt{Km}$ from the radar. The path loss term is $\psi = 3.5$. The CPI length $M_{\texttt{CPI}}$ is $400$ pulses. Each PRI, actions with $D(w_{t}, w_{t-1}) > (\hat{d} = 0.2)$ are eliminated.

\begin{figure}
	\centering
	\includegraphics[scale=0.45]{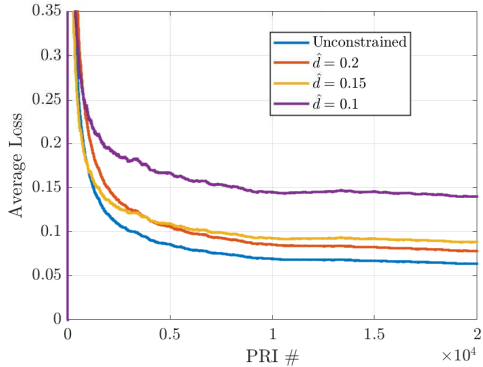}
	\caption{\textsc{Impact of maximum tolerable distortion level} $\hat{d}$ on the average cost incurred by the Thompson Sampling algorithm.}
	\label{fig:dLevel}
	\vspace{-.4cm}
\end{figure}

\begin{figure*}
	\centering
	\begin{subfigure}[]
		\centering
		\includegraphics[scale=0.45]{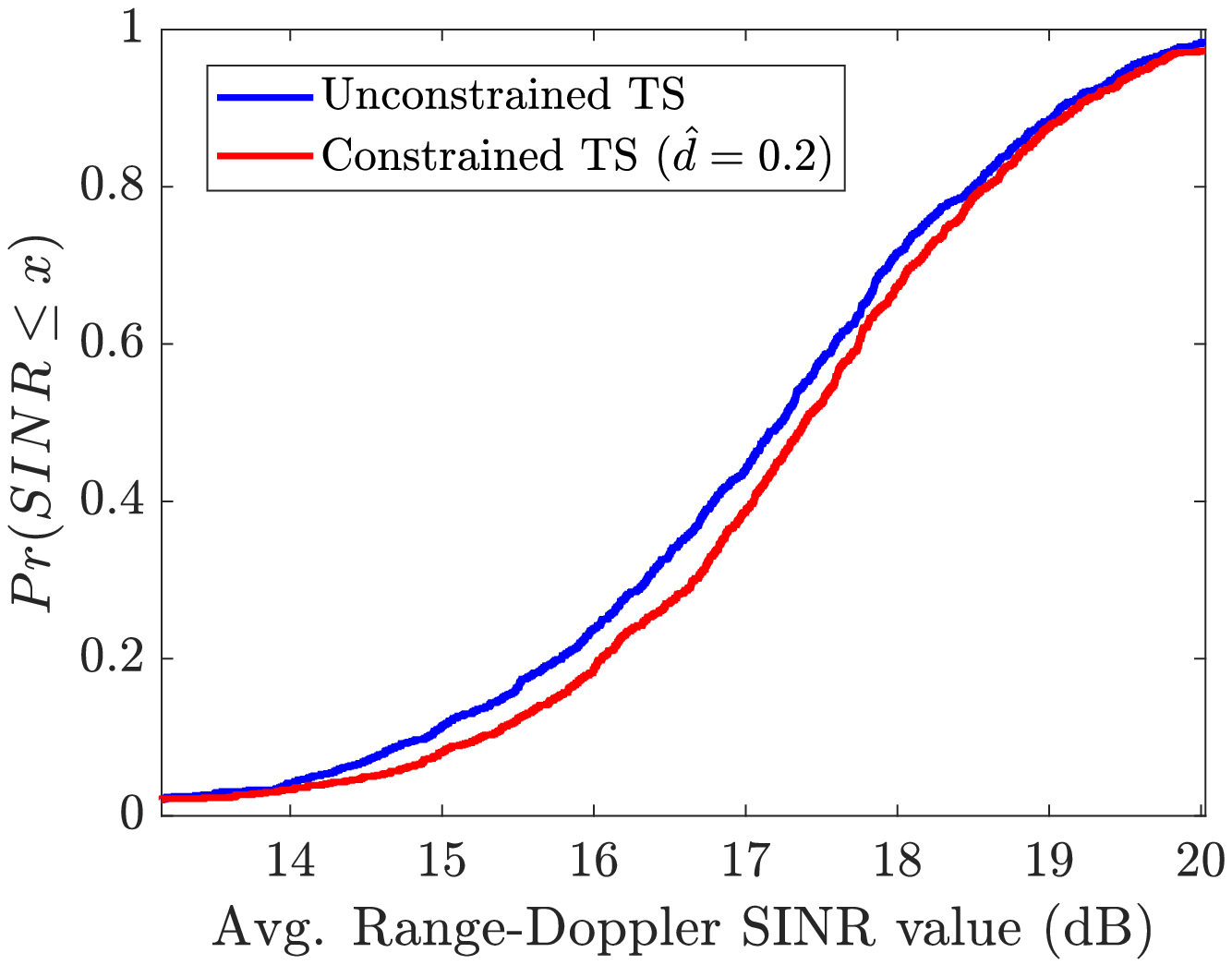}
	\end{subfigure}
	\begin{subfigure}[]
		\centering
		\includegraphics[scale=0.45]{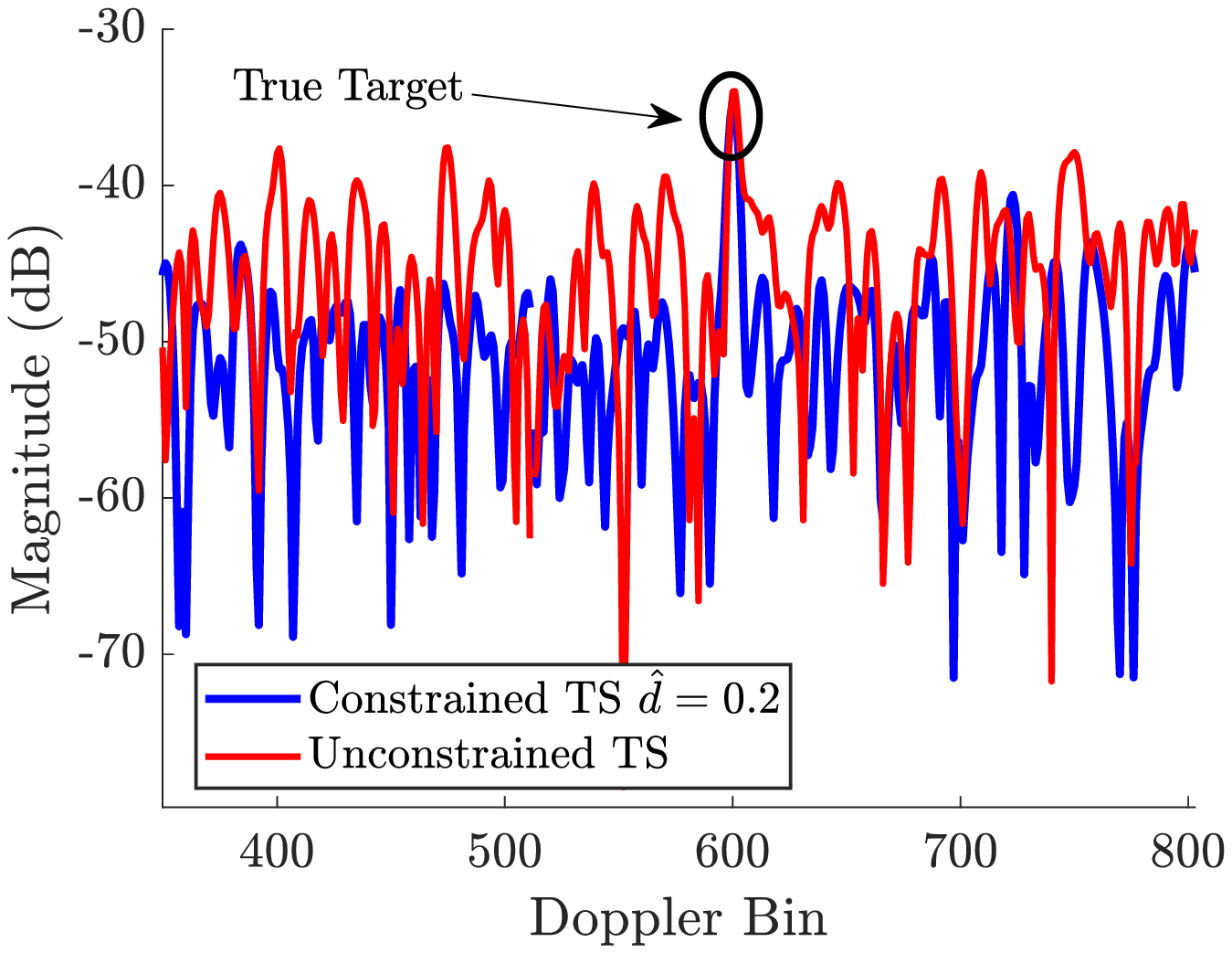}
	\end{subfigure}
	\caption{(a) \textsc{Empirical cumulative distribution function} of the SINR experienced by the constrained TS algorithm with tolerable distortion level $\hat{d}=0.2$ and unconstrained TS algorithm. (b) \textsc{Doppler profile} of the constrained and unconstrained TS algorithms at the true target range bin.}
	\label{fig:ecdf}
\end{figure*}

\begin{figure*}
	\centering
	\includegraphics[scale=0.35]{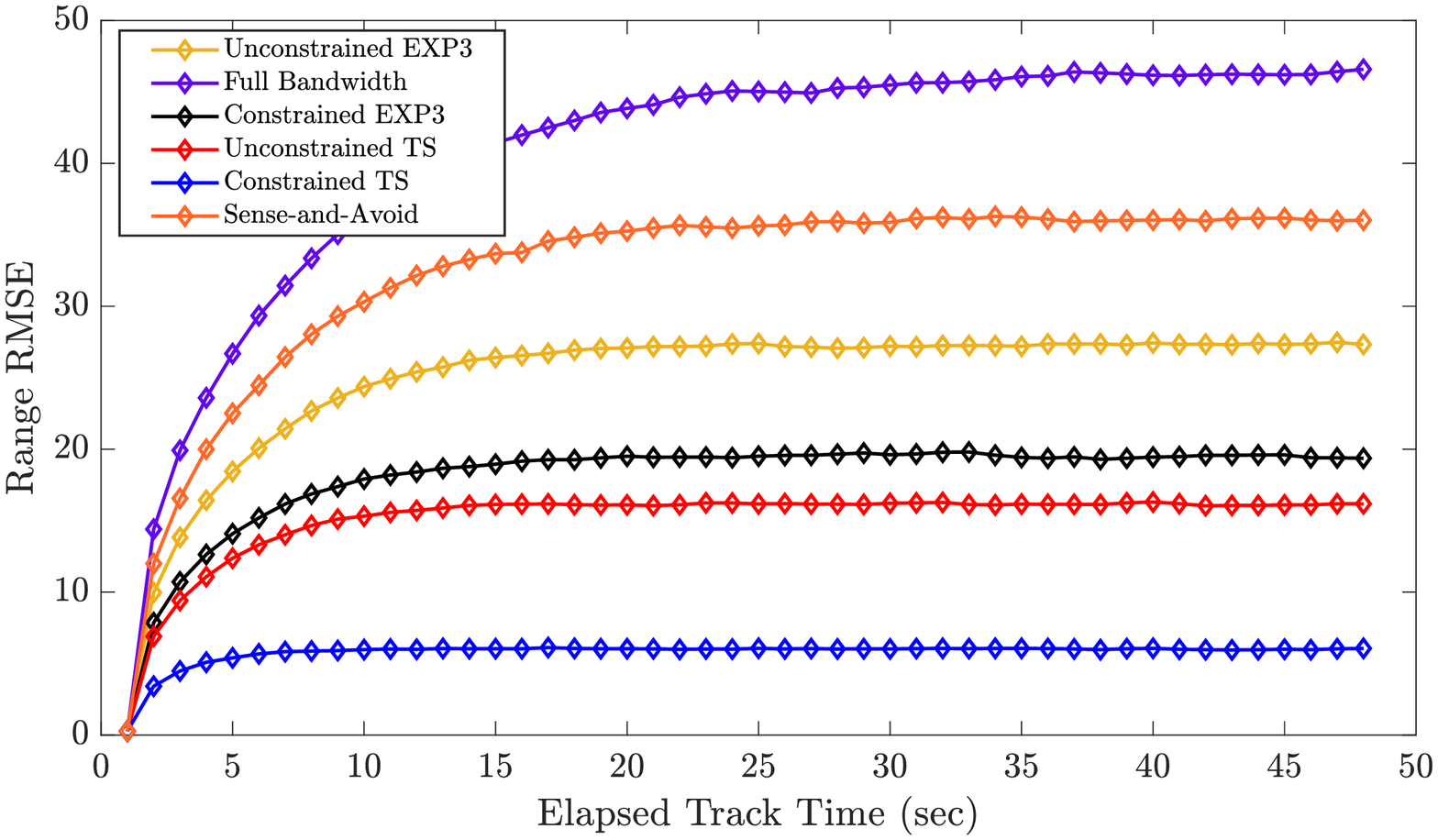}
	\caption{\textsc{Kalman filtered tracking RMSE} of each learning algorithm in the coexistence setting.}
	\label{fig:coexistTrack}	
\end{figure*}

In Figure \ref{fig:tc7Detect}(a), the average cost of each online learning algorithm is evaluated over 20,000 PRIs. In this scenario, the assumption of a stationary stochastic environment holds well, and both the unconstrained and constrained variants of the TS approach outperform EXP3. Additionally, we observe that utilizing a constraint of $\hat{d}=0.2$ results in a slightly worse long-term average cost for both the TS and EXP3 algorithms, indicating that utilizing the time-varying waveform catalog does inhibit the radar's ability to avoid interference. A more detailed examination of how $\hat{d}$ effects the radar's ability to avoid interference can be seen in Figure \ref{fig:dLevel}. Once actions which are critical to either exploration or exploitation are removed, the radar performs significantly worse, as observed for the case of $\hat{d} = 0.1$ in Figure \ref{fig:dLevel}.

\begin{figure*}[t]
	\centering
	\begin{subfigure}[]
		\centering
		\includegraphics[scale=0.35]{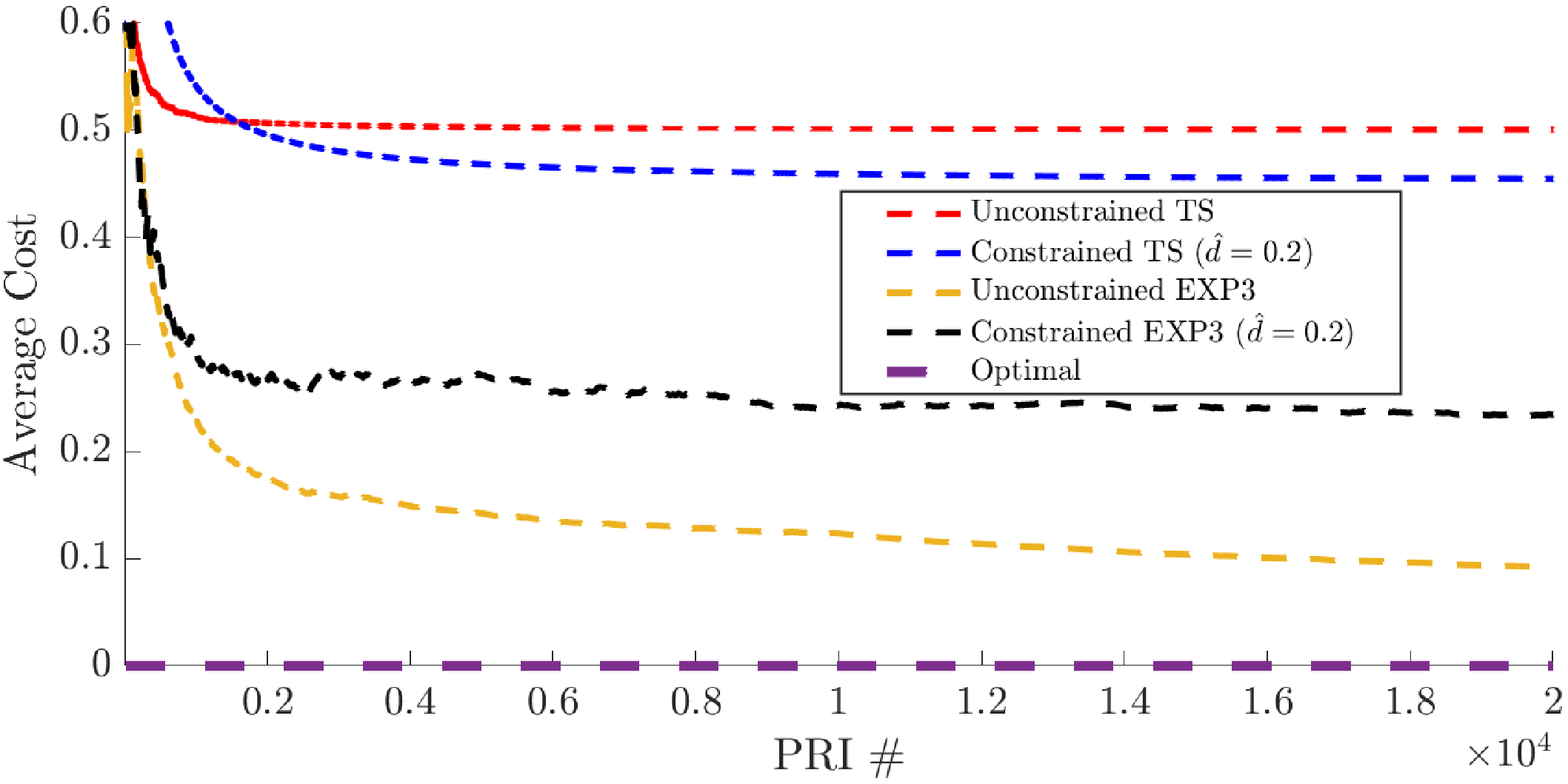}
	\end{subfigure}
	\begin{subfigure}
		\centering
		\includegraphics[scale=0.35]{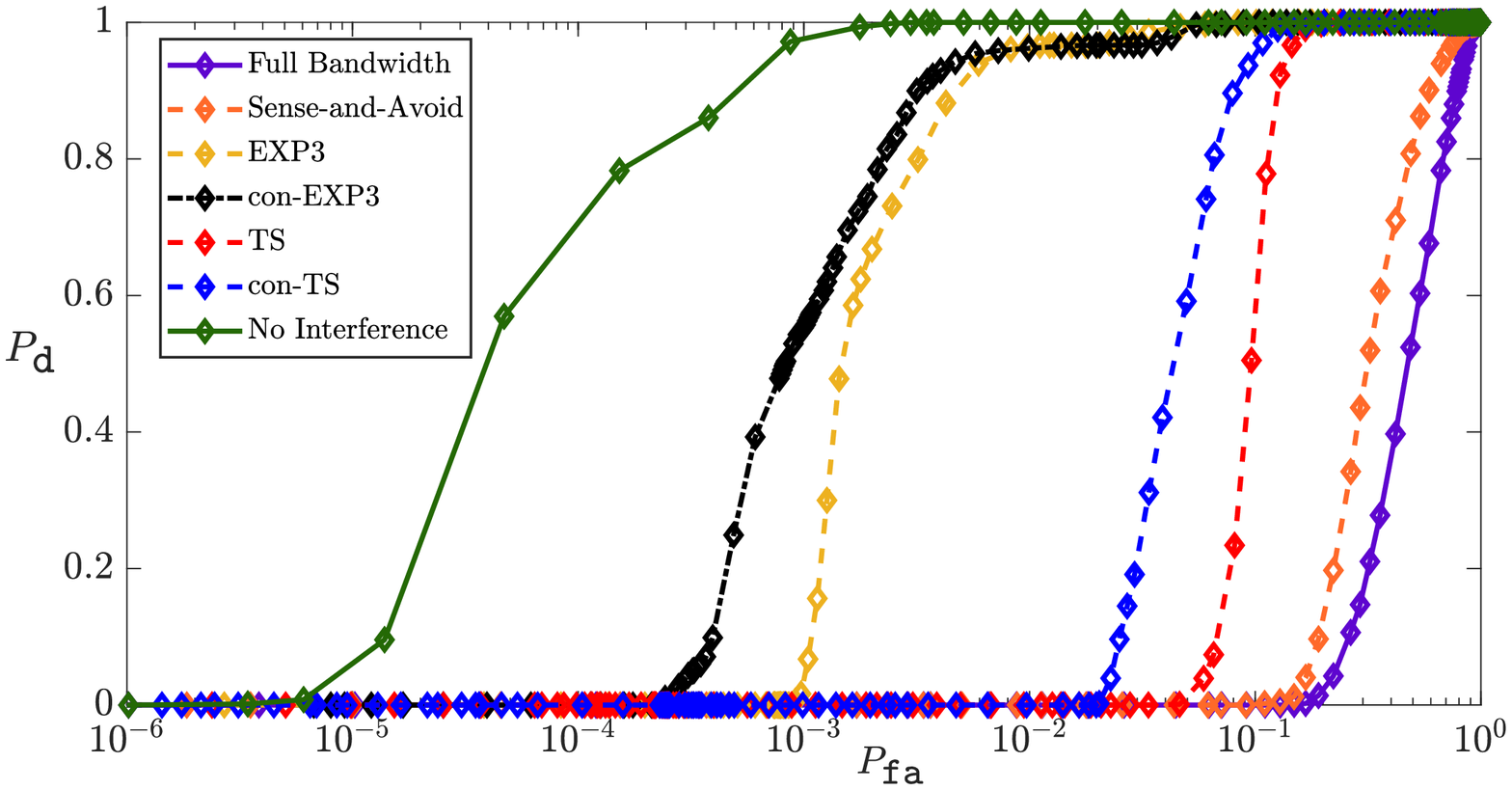}
	\end{subfigure}
	\caption{(a) \textsc{Average cost} of online learning algorithms over $n=20,000$ PRIs in the intentional adaptive jamming scenario. (b) \textsc{Average detection characteristics} of each online learning algorithm compared to a static radar which occupies the entire shared channel.}
	\label{fig:jamm}
\end{figure*}

In Figure \ref{fig:tc7Detect}(b), the detection performance of each online learning algorithm is evaluated and compared to the case of a statically allocated radar occupying the entire 100$\texttt{MHz}$ shared channel. In contrast to the interference avoidance performance, we observe that the constrained variants of each algorithm substantially outperform their unconstrained counterparts in terms of detection characteristics. This is due to the reduced sidelobe levels in the resulting range-Doppler images due to the constrained waveform set. By accepting a slightly higher rate of collision with interference, the radar is able to adapt its waveform less drastically, improving the quality of the range-Doppler image. Further, it is noted that each of the online learning algorithms provides an improvement in detection over a naive sense-and-avoid scheme, which selects the largest contiguous bandwidth not occupied by the previous interference state $\mathbf{s}_{t-1}$. The reasons for this improvement are twofold. First, the contextual bandit algorithms are able to identify the underlying statistical behavior of the channel while the sense-and-avoid scheme acts only according to the previous interference state, resulting in favorable interference avoidance while exhibiting less waveform agility. Further, the contextual bandit algorithms may choose to entirely avoid subsets of the shared channel where interference occurs often, while the sense-and-avoid scheme seeks to utilize the largest available bandwidth in all scenarios. 

In Figure \ref{fig:ecdf}(a), we see the cumulative distribution function of the average range-Doppler image SINR for the constrained and unconstrained TS algorithms. Even though the constrained variant was seen to perform slightly worse than its unconstrained counterpart in terms of interference avoidance, the average SINR across the range-Doppler map is sightly improved due to decreased spreading of energy across the Doppler domain as a result of pulse-agility. This is confirmed by Figure \ref{fig:ecdf}(b), which shows the Doppler profile of the constrained and unconstrained TS algorithms respectively at the true target bin. While the target peak is visible in both cases, the Doppler sidelobes are much higher when utilizing the unconstrained adaptive approach, explaining the disparity in detection performance.

In Figure \ref{fig:coexistTrack}, the tracking performance of each approach is seen when using a Kalman filter and the optimal bandwidth selection described in the previous section. As expected, the constrained Thompson sampling approach results in the lowest cumulative RMSE due to the combination of interference avoidance and the pronouced target peak as a result of the distortion constraint. Despite the modest gain in average SINR, the constrained TS algorithm performs much better than its unconstrained counterpart due to the decreased number of false targets that appear due to waveform agility. Further, each learning algorithm results in a marked improvement over the conventional fixed bandwidth and naive reactive radars, demonstrating the impact of online learning even when the interference avoidance performance is suboptimal, as is the case for the EXP3 algorithms in this setting. 

\begin{figure}[t]
	\centering
	\includegraphics[scale=0.45]{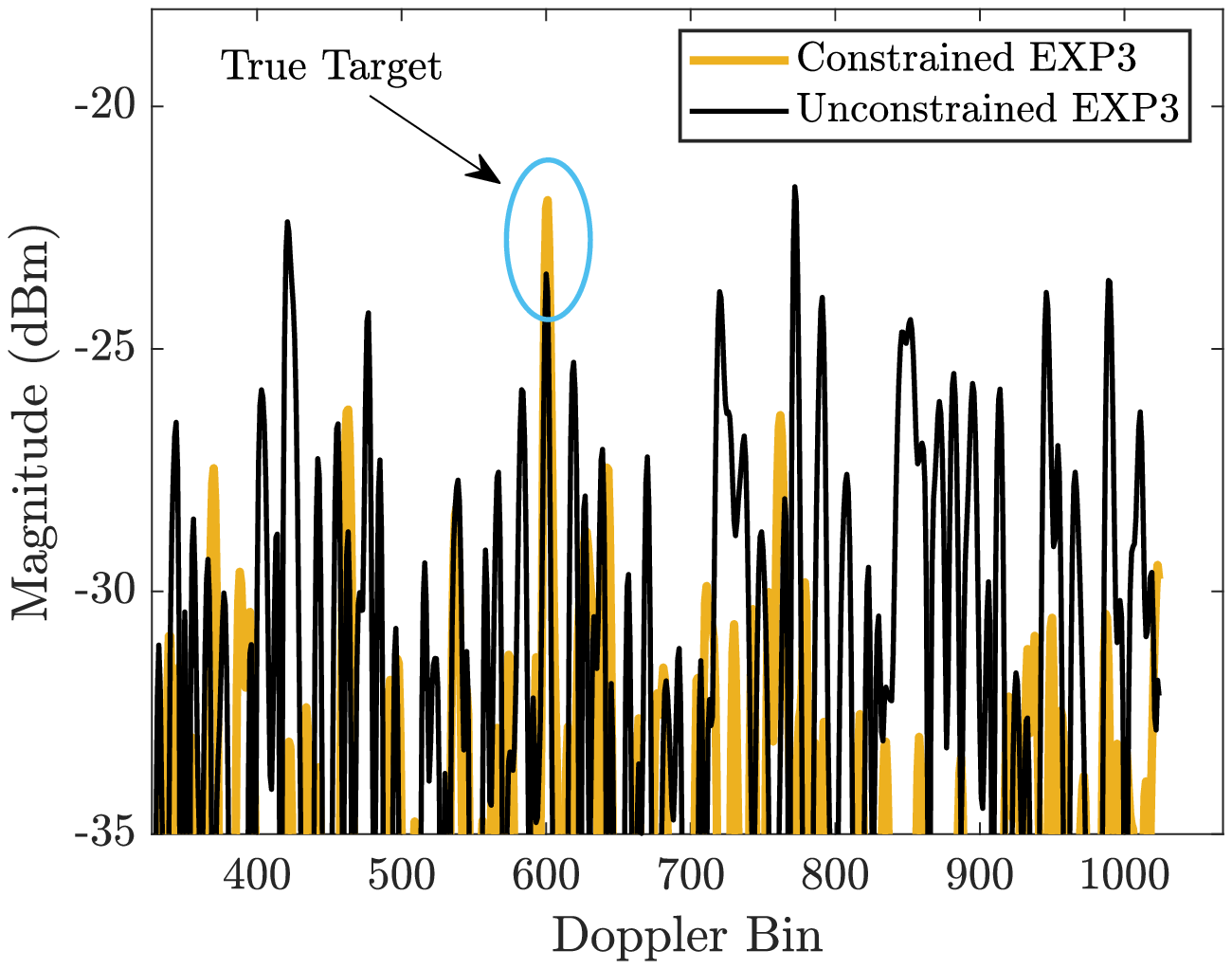}
	\caption{\textsc{Doppler profile} of unconstrained EXP3 algorithm and constrained EXP3 with distortion constraint $\hat{d}=0.2$. When the distortion constraint is implemented, the true target peak is clearly visible. Without the distortion constraint, energy is spread across the Doppler domain and the sidelobe levels mask the target peak.}
	\label{fig:jamDop}
\end{figure}

\subsection{Intentional Adaptive Jamming Scenario}
\begin{figure*}
	\centering
	\includegraphics[scale=0.35]{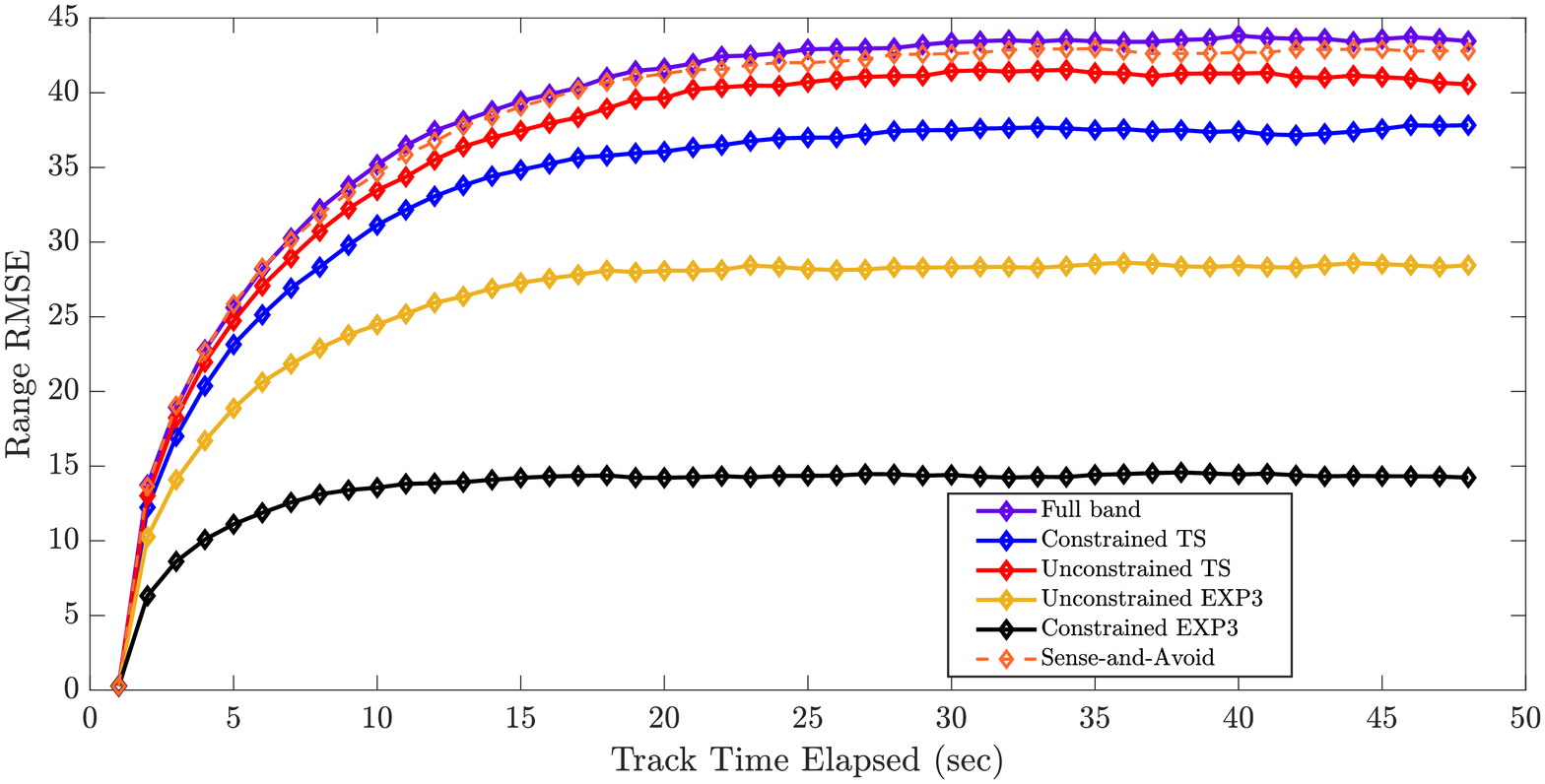}
	\caption{\textsc{Tracking RMSE} of each approach in the adaptive jamming scenario.}
	\label{fig:jamRMSE}
\end{figure*}

In this environment, the radar must avoid interference from an adaptive jammer. If the radar selects the same waveform for two consecutive PRIs, $w_{t} = w_{t-1}$, then during the next PRI a high-powered interfering signal occupies the bandwidth utilized by $w_{t}$. If the radar adapts its waveform, $w_{t} \neq w_{t-1}$, the jammer occupies the bandwidth it utilized in the previous PRI. This is a more challenging scenario for the CR than the coexistence environment as the radar's actions influence the behavior of the interference and the jammer may adapt every PRI. In the following simulation, we consider a constant jammer-to-noise-ratio (JNR) of $20 \texttt{dB}$ at the radar when interference occurs. As in the coexistence setting, actions with $D(w_{t},w_{t-1}) > (\hat{d} = 0.2)$ are eliminated.

In Figure \ref{fig:jamm}(a), the average cost of each learning algorithm is seen over $n = 20,000$ PRIs. Due to the nonstationary nature of the interference process, both the constrained and unconstrained TS algorithms perform substantially worse in terms of average cost in this setting relative to the coexistence scenario, which is expected given that the stochastic bandit model assumes a stationary mapping between contexts and costs exists. It is observed that both the constrained and unconstrained variants of EXP3 perform substantially better than TS, driving the long term average cost below $0.1$ in the unconstrained case. Although implementing the distortion constraint of $\hat{d} = 0.2$ results in worse interference avoidance relative to unconstrained EXP3, we see from Figure \ref{fig:jamm}(b) that the detection performance of both EXP3 and TS is substantially improved by the distortion constraint in this setting. To see this, a typical Doppler profile for this scenario is shown in Figure \ref{fig:jamDop}. The Doppler sidelobe levels are seen to be much lower when a distortion constraint of $\hat{d} = 0.2$ is implemented, and a higher peak is seen at the location of the true target.

In Figure \ref{fig:jamRMSE}, the performance of the radar's tracker is seen over 20 seconds. As expected from the detection characteristics, the constrained EXP3 algorithm provides superior performance in terms of range tracking. This is presumably due to the more pronounced target peak than in the case of the unconstrained algorithm, which results in a more accurate estimation of range and Doppler values and reduces the variance of noisy measurements fed to the Kalman tracker.

\section{Conclusion and Discussion}
\label{se:concl}
The cognitive radar waveform selection process was formulated as a linear contextual bandit problem with a time-varying constrained action set. An online learning framework for selecting optimal waveforms in both stochastic and more general adversarial environments was developed, resulting in a computationally feasible structure and performance guarantees. To mitigate Doppler sidelobes associated with pulse-to-pulse waveform adaptation, we introduced a time-varying constraint on the radar's waveform catalog. We then discussed how the proposed waveform selection process could interact with a tracking sub-system to improve overall performance. Finally, the proposed scheme was numerically evaluated in a radar-cellular coexistence scenario as well as in the presence of intentional adaptive jamming. Simulation results demonstrate that the proposed online learning algorithms provide favorable performance in terms of interference mitigation, target detection, and tracking error minimization compared to traditional fixed bandwidth radar and a naive adaptive scheme.

While this investigation was limited to the class of LFM waveforms, the approach is applicable to broader waveform classes. Future work could focus on utilizing a similar online learning approach to select optimal codes for phase-coded waveforms, and incorporate the ambiguity function into the waveform selection process. Additionally, we have focused on simple point targets in interference-limited scenarios. Future work could examine a broader range of physical environments where the characteristics may drift over time.

Other work in this domain could focus on developing further connections between online learning algorithms and particle filtering, which could be used to directly optimize the tracker. Another possible extension could explore a distributed sensing application as a multi-agent system, where individual sensor-processor agents could cooperate or compete to improve performance or facilitate interoperability. Given the necessity of intelligent and adaptive sensing for emerging applications, it is likely that online learning algorithms will continue to play an important role in the development of future radars.

\appendix
\subsection{Proof of Proposition 1}
Since $\mathcal{C}(w_{i},\mathbf{s}_{t})$ is a linear combination of two terms, $\beta_{1} \texttt{BW}_{c}(w_{i},\mathbf{s}_{t})$ and $\beta_{2} \texttt{BW}_{miss}(w_{i},\mathbf{s}_{t})$, it is sufficient to show each term is locally Lipschitz continuous for any pair of waveforms $(w_{i}, w_{j}) \in \mathcal{W} \times \mathcal{W}$, since the sum of locally Lipschitz functions is locally Lipschitz. The first term considers the overlap with a fixed interference vector $\mathbf{s}_{t}$. Let the Lipschitz metric be the distortion function
\begin{equation}
D(w_{t},w_{t-1}) \triangleq \gamma_{1} \lVert f_{t} - f_{t-1} \rVert^{2} + \; \gamma_{2} \lVert \texttt{BW}_{t} - \texttt{BW}_{t-1} \rVert^{2} 
\end{equation}
where $\gamma_{1}, \gamma_{2} \geq 0$ are fixed constants. Thus, waveforms with $D(w_{i},w_{j}) \rightarrow 0$ will yield similar values of $\texttt{BW}_{c}$ due to increasing similarity in spectral overlap with the fixed $\mathbf{s}_{t}$. The second term is also dependent on the frequency content of each waveform, so a similar argument holds for $\texttt{BW}_{miss}(w_{i},\mathbf{s}_{t})$. \qed

\subsection{Proof of Proposition 2}
An analysis of the Bayesian (expected) regret of linear Thompson Sampling for contextual bandits was first studied by Russo and Van Roy in \cite{infoTS}. From proposition 3 of \cite{infoTS}, we note the Bayes regret $\mathcal{O}(d \log{n} \sqrt{n})$, where $d$ is the dimensionality of the context vector, holds for the waveform selection problem due to $|\mathcal{W}| < \infty$, the subgaussian disturbance assumption, and noting the constrained action set $\mathcal{W}'$ is a special case of the stochastic decision set $\mathcal{A}_{t}$ proposed in \cite{infoTS}. Setting $d = 3$, we have (\ref{eq:bayesTS}).

The worst-case (frequentist) regret of TS was studied by \cite{agrawal}. By Theorem 1 of \cite{agrawal}, and noting that the constant exploration parameter $\epsilon$ can be set to $\frac{1}{\log{n}}$ with a constant problem dimensionality $d = 3$, we have (\ref{eq:freqTS}). We further note that the constrained action set $\mathcal{W}'$ is consistent with the problem considered in \cite{agrawal}, as we can formulate a problem where the context vector is $\mathbf{0}_{d}$ for any waveform in $\mathcal{W}$ that is not in $\mathcal{W'}$. \qed

\subsection{Proof of Proposition 3}
We begin by assuming a learning rate $\varepsilon$ such that $\varepsilon \; \hat{C}_{t}(w_{i},\hat{\mathbf{s}}_{t}) \geq 0$, $\forall w_{i} \in \mathcal{W}, \mathbf{s}_{t} \in \mathscr{S}$. Applying Theorem 11.1 of \cite{Lattimore2020}, we have the bound
\begin{multline}
	\texttt{Regret}(n) \leq \frac{\log{W}}{\varepsilon} + 2 \gamma n + \varepsilon \sum_{t=1}^{n} \\ \E \left[\sum_{w_{i} \in \mathcal{W'}, \mathbf{s}_{t} \in \mathscr{S}} \tilde{P}_{t}(w_{i}) \hat{C}_{t}(w_{i}, \mathbf{s}_{t})^{2} \right],
\end{multline}
where $\gamma \in (0,1)$ is the exploration parameter of Algorithm \ref{algo:exp3algo}. Setting $\gamma \leq nd$ and using the fact that $\E[\sum_{w_{i} \in \mathcal{W'}, \mathbf{s}_{t} \in \mathscr{S}} \tilde{P}_{t}(w_{i}) \hat{C}_{t}(w_{i}, \mathbf{s}_{t})^{2}|\tilde{P}_{t}] \leq d$ (seen in the proof of Theorem 27.1 of \cite{Lattimore2020}), we are left with
\begin{equation}
	\texttt{Regret}(n) \leq \frac{\log(W)}{\varepsilon} + \varepsilon n (3d).
\end{equation}

Choosing $\varepsilon = \frac{log(W)}{3dn}$ and setting the context dimension $d = 3$ we arrive at (\ref{eq:exp3bound}) \qed

\bibliography{tsbib}{}
\bibliographystyle{IEEEtran}

\end{document}

%% file: notation.tex
\def\nbf{{\mathbf{f}}}

\def\nbh{{\mathbf{h}}}

\def\nbv{{\mathbf{v}}}

\def\nbx{{\mathbf{x}}}

\def\nbz{{\mathbf{z}}}
\def\nb0{{\mathbf{0}}}
\def\nb1{{\mathbf{1}}}

\def\nbX{{\mathbf{X}}}

\def\nbZ{{\mathbf{Z}}}

\def\nbbN{{\mathbb{N}}}

\def\nbbR{{\mathbb{R}}}

\def\argmin{\operatorname{arg~min}}
\def\argmax{\operatorname{arg~max}}
\def\E{\mathbb{E}}

\def\P{\mathbb{P}}

